\documentclass[11pt,a4paper]{article}

\usepackage[margin=1in]{geometry}
\usepackage{amsmath,amssymb,amsfonts,amsthm,mathtools,bm}
\usepackage{braket}
\usepackage{booktabs}
\usepackage{array}
\usepackage{tabularx}
\usepackage{graphicx}
\usepackage{microtype}
\usepackage{enumitem}
\usepackage{tikz}
\usetikzlibrary{arrows.meta,positioning,calc,decorations.pathreplacing}
\usepackage[colorlinks=true,linkcolor=blue,citecolor=blue,urlcolor=blue]{hyperref}
\usepackage{float}

\DeclareMathOperator{\Tr}{Tr}
\DeclareMathOperator{\supp}{supp}
\DeclareMathOperator{\sech}{sech}
\DeclareMathOperator{\artanh}{artanh}

\newcommand{\dd}{\mathrm d}
\newcommand{\ii}{\mathrm i}

\newcommand{\I}{\mathbb I}
\newcommand{\norm}[1]{\lVert #1\rVert}

\newcommand{\Frob}{\mathrm F}
\newcommand{\HS}{\mathrm{HS}}
\newcommand{\Cqp}{C_{\rm qp}}

\newtheorem{theorem}{Theorem}[section]
\newtheorem{proposition}[theorem]{Proposition}

\theoremstyle{definition}

\theoremstyle{remark}
\newtheorem{remark}[theorem]{Remark}

\begin{document}

\title{Infrared Memory of a Majorana Shutter}
\author{Ali Vahedi\thanks{Email: \texttt{vahedi@khu.ac.ir}}\\
\small Department of Physics, Kharazmi University, Tehran, Iran}
\date{October 2026}
\maketitle

\begin{abstract}
We study a finite-width Majorana shutter as a local, time-dependent scattering defect and determine which features of its infrared memory are insensitive to the microscopic spatial profile. For a smooth shutter, the zero-frequency scattering matrix is fixed by the integrated coupling, while finite-width corrections vanish in the long-wavelength limit. A time-dependent switch then produces an infrared Bogoliubov kernel whose singular part is controlled by the change in this scattering matrix.

We examine the many-body response through three observables. The overlap of the pre- and post-shutter Gaussian vacua exhibits Anderson orthogonality, with an exponent set by the scattering angle. A lattice calculation of the post-quench quasiparticle number resolves its logarithmic infrared growth, reaching the predicted coefficient in the slow-switch test and reproducing the expected dependence on the scattering-angle change. When the entanglement cut passes through the shutter, the entropy follows the effective-central-charge description of a conformal defect, with parameter-free agreement at the percent level, improving for the wider shutter. When the shutter lies inside the interval, its contribution remains a subleading finite-size correction. A local bilinear correlator approaches the expected Majorana power law and shows no stable logarithmic amplitude.

These results show that the same low-frequency scattering data control distinct global and spatial observables, while their infrared coefficients are not interchangeable. The remaining limits are also made explicit: the finite-memory correction to the continuum kernel and the controlled continuum limit of large-interval spatial entropy.
\end{abstract}

\section{Introduction}
\label{sec:intro}

A massless Majorana field in one spatial dimension naturally separates
into right- and left-moving components. We use $c$ for the conformal
central charge of a critical theory; a single Majorana conformal field
theory has $c=1/2$ \cite{Majorana1937,CalabreseCardy2004}. If a localized time-dependent
interaction mixes those channels, the device acts as a \emph{chiral
shutter}: the asymptotic channel identification before the switch
differs from that after the switch. The simplest low-frequency
description is an orthogonal two-channel map,
\begin{equation*}
\boldsymbol\chi_{\rm out}(t)=S(t)\boldsymbol\chi_{\rm in}(t),
\qquad S(t)\in SO(2),
\end{equation*}
with the $SO(2)$ statement understood at zero frequency (and as the
corresponding low-frequency limit). At finite frequency the scattering
matrix obeys the Majorana reality constraint and unitarity but is
generally complex. When the early- and late-time rotations differ, the
positive-frequency decomposition changes. The corresponding Bogoliubov
transformation contains an anomalous block and develops infrared memory.
In the idealized point-contact description this memory appears as a
logarithmically divergent squared Hilbert--Schmidt norm of the anomalous
kernel \cite{Vahedi2026a,Araki1971,ShaleStinespring1965}.

That global result raises a different question from the scattering
problem itself: \emph{where does the entanglement generated by the
shutter live?} A frequency-space bipartition can directly isolate the
soft modes responsible for the infrared obstruction. A spatial interval
instead probes a local algebra \cite{Haag1996}. There is no general reason for a global
non-implementability statement to determine the entropy of a bounded
region.

The problem is also related in spirit to time-dependent boundary
problems and the moving-mirror literature, where a time-dependent
boundary condition changes the relation between in- and out-vacua and
can produce particle creation and nontrivial stress-tensor expectation
values \cite{FullingDavies1976,DaviesFulling1977}. Our shutter is not
literally a moving mirror: it is an internal two-channel Majorana
mixing problem. The analogy is useful because it places the
time-dependent positive-frequency decomposition in a broader
field-theoretic context, while the finite-width construction here keeps
the microscopic scattering problem explicit.

More directly, the shutter problem is a two-channel analogue of
entanglement across conformal defects in one-dimensional critical
systems. For a static conformal defect in a free-fermion chain with
transmission amplitude $s$, the entanglement entropy of a subregion
containing the defect acquires an effective central charge
$c_{\rm eff}(s)$ that interpolates between $0$ (perfectly reflecting
defect) and $c$ (perfectly transmitting defect)
\cite{EislerPeschel2010,Peschel2012}. The corresponding
time-dependent problem, where two critical half-chains are suddenly
coupled through a defect, exhibits logarithmic growth of entanglement
governed by the same $c_{\rm eff}$ \cite{EislerPeschel2012}. The
shutter considered here is not identical to a conformal defect --- it
is time-dependent with a finite width and a smooth switch --- but the
connection is close enough that the defect literature provides the
natural language for interpreting the weak-shutter coefficient
$Q_\alpha$ introduced below. We make this connection explicit in
Sec.~\ref{sec:defect-interpretation}.

There are therefore two distinct universality questions. The first is
a one-particle scattering question: whether different microscopic
profiles with the same integrated coupling approach the same
low-frequency scattering matrix. The second is an entanglement
question: whether that one-particle universality survives the sequence
\begin{equation}
S_a\longrightarrow\beta\longrightarrow\Gamma\longrightarrow\Gamma_A\longrightarrow S_A^{(\alpha)}.
\label{eq:logic-chain}
\end{equation}
The first question is answered analytically below. The second requires
additional control after spatial or mode restriction and is therefore
treated separately.

The distinction is particularly sharp for a Majorana field. The vacuum
already has a universal spatial interval entropy determined by the
Majorana conformal field theory (CFT), whose central charge is $c=1/2$,
\begin{equation}
S_A^{(0)}=\frac16\log\frac{L}{a_{\rm UV}}+O(1),
\label{eq:vacuum-CFT}
\end{equation}
where $L$ is the interval length and $a_{\rm UV}$ is a short-distance
regulator \cite{CalabreseCardy2004}. The quantity of interest is
therefore the shutter-induced difference
\begin{equation}
\Delta S_A(a,L)=S_A(a,L)-S_A^{(0)}(L).
\label{eq:deltaS-intro}
\end{equation}
The same covariance spectrum also gives all Rényi entropies, allowing
the entanglement spectrum rather than only its von Neumann functional to
be tested.

The microscopic regularization considered here is a smooth finite-width
coupling,
\begin{equation}
H_{\rm int}(t)=\ii\int\dd x\,\lambda(t)w_a(x)\chi_R(x,t)\chi_L(x,t),
\qquad
w_a(x)=a^{-1}w(x/a),
\label{eq:width-intro}
\end{equation}
with $w$ smooth, compactly supported, even, nonnegative, and normalized
to unit integral. Here $a$ denotes the physical width of the shutter,
whereas $a_{\rm UV}$ denotes the ultraviolet spatial regulator used when
a continuum interval entropy is discussed. They are independent scales.
This construction has two useful features. First, its zero-frequency
transfer matrix can be evaluated exactly. Second, its finite-frequency
corrections vanish when the interaction becomes short compared with the
wavelength. These properties make the finite width a natural test of
whether a scattering universality class remains visible in entanglement
observables.

The central logical distinction of this paper is
\begin{equation}
\boxed{
\text{global infrared non-implementability}
\;\not\Rightarrow\;
\text{a determined spatial infrared entropy divergence}.
}
\label{eq:central-distinction}
\end{equation}
For frequency partitions the infrared sector can be isolated directly.
For a spatial interval, however, the relevant object is the covariance
after spatial restriction, and that extra operation can change the way
the infrared singularity enters the entropy.

We make six contributions. First, we give a self-contained finite-width
scattering construction, including the transfer matrix, its
zero-frequency limit, and the relation between the finite-width coupling
and the symmetric point-interaction parametrization. Second, we
formulate the Gaussian entanglement problem in a way that treats
anomalous pairing correctly and gives the complete Rényi family. Third,
we implement three partitions --- the two lattice chains, soft/hard
energy modes, and space --- and report a finite-lattice calculation for
several profiles. The lattice regulator is explicitly identified as a
doubled Majorana realization rather than a one-to-one discretization of
a single chiral continuum sector. Fourth, we formulate the continuum
extrapolation needed to turn the pilot calculation into an entanglement
universality test, and we state precisely which operator estimate is
missing. Fifth, we measure the many-body orthogonality exponent and
Loschmidt scaling, and use a controlled connected two-point-function
analysis as a consistency check on the ordinary free-fermion exponent.
Sixth, we make the connection to the conformal-defect literature
explicit and identify the shutter transmission amplitude at zero
frequency.

The paper is intentionally conservative about the continuum entropy
question. The numerical data establish that the finite-width Gaussian
problem can be evaluated and that its entanglement spectrum has a
nontrivial and partition-dependent response. They do not establish a
continuum large-$L$ exponent. This distinction is essential because
entropy is a nonlinear spectral functional, whereas the microscopic
universality statement is a one-particle scattering statement.

\section{Finite-width chiral shutter}
\label{sec:model}

\subsection{Free Majorana channels}

We consider two real chiral fields satisfying
\begin{equation}
(\partial_t+\partial_x)\chi_R=0,
\qquad
(\partial_t-\partial_x)\chi_L=0.
\label{eq:free-eom}
\end{equation}
The free Hamiltonian may be written, up to the conventional
normal-ordering prescription, as
\begin{equation}
H_0=\frac{\ii}{2}\int\dd x\,
\left(\chi_R\partial_x\chi_R-\chi_L\partial_x\chi_L\right).
\label{eq:H0}
\end{equation}
The equal-time normalization is
\begin{equation}
\{\chi_i(x,t),\chi_j(y,t)\}=\delta_{ij}\delta(x-y).
\label{eq:field-CAR}
\end{equation}
Because the fields are real, it is sufficient to integrate over
positive frequencies; the creation-operator terms supply the
negative-frequency components required by the Majorana reality
condition. A convenient positive-frequency expansion is
\begin{align*}
\chi_R(x,t)&=\int_0^\infty\frac{\dd\omega}{\sqrt{2\pi}}
\left[a_R(\omega)e^{-\ii\omega(t-x)}+a_R^\dagger(\omega)e^{\ii\omega(t-x)}\right],\\
\chi_L(x,t)&=\int_0^\infty\frac{\dd\omega}{\sqrt{2\pi}}
\left[a_L(\omega)e^{-\ii\omega(t+x)}+a_L^\dagger(\omega)e^{\ii\omega(t+x)}\right],
\end{align*}
with
\begin{equation*}
\{a_i(\omega),a_j^\dagger(\nu)\}=\delta_{ij}\delta(\omega-\nu),
\qquad
\{a_i(\omega),a_j(\nu)\}=0.
\end{equation*}

The coincident product in Eq.~\eqref{eq:width-intro} is understood
through the associated one-particle Dirac operator together with a UV
regulator, or equivalently by normal ordering. The scattering problem
below is finite-dimensional and does not require manipulating products
of operator-valued distributions directly.

\subsection{Profile class and finite width}

We assume
\begin{equation}
\int_{-\infty}^{\infty}w(u)\,\dd u=1,
\qquad
\supp w\subset[-1,1],
\qquad
w(u)=w(-u),
\qquad
w(u)\ge0.
\label{eq:profile}
\end{equation}
Then $w_a$ has support in $[-a,a]$ and converges weakly to a delta
function. The evenness assumption is convenient because it eliminates
a first-moment contribution to the low-frequency expansion. The profile
class is deliberately narrower than the set of all localized
interactions; sign-changing or resonant profiles can carry additional
low-energy structure and are not claimed to belong to the same
universality class.

For all numerical calculations we use three normalized members of this
class. With the normalization in Eq.~\eqref{eq:profile},
\begin{align}
 w_1(u)&\propto \exp[-1/(1-u^2)],\nonumber\\
 w_2(u)&\propto (1+u^2/2)\exp[-1/(1-u^2)],\nonumber\\
 w_3(u)&\propto (1-u^2)^4,
\label{eq:numerical-profiles}
\end{align}
for $|u|<1$, and all three profiles vanish outside that interval. The
proportionality constants are fixed by Eq.~\eqref{eq:profile}. The
labels $w_1,w_2,w_3$ used in the numerical tables therefore refer to
these specific normalized spatial profiles, not to different coupling
strengths.

\subsection{Time-dependent switch}

The coupling approaches constants at early and late times,
\begin{equation*}
\lambda(t)\to\lambda_-,\qquad t\to-\infty,
\qquad
\lambda(t)\to\lambda_+,\qquad t\to+\infty.
\end{equation*}
We use two switching protocols. By a \emph{sudden switch} we mean the
idealized instantaneous quench $\lambda(t)=\lambda_f\Theta(t)$,
used below to define the post-quench excitation content and the
ground-state overlap. By a \emph{smooth switch} we mean a finite-duration
change, for which the numerical protocol is
\begin{equation}
\lambda(t)=\frac{\lambda_f}{2}\left[1+\tanh(t/\tau)\right]
\label{eq:tanh-switch}
\end{equation}
so that $\lambda_-=0$ and $\lambda_+=\lambda_f$. Here $\tau$ is
the switching time. The static calculation determines the response
at fixed coupling, whereas the switched problem asks what happens when
that coupling changes while a wave packet traverses the interaction
region. We use the term ``sudden'' rather than ``sharp'' to distinguish
this idealized temporal quench from the finite spatial width $a$.
Two time scales therefore enter: the transit time $a$ (velocity set to
one) and the switching time $\tau$. We define the memory parameter
\begin{equation}
\epsilon_{\rm mem}\equiv\frac{a}{\tau}.
\label{eq:memory-parameter}
\end{equation}
The exact finite-width time-dependent scattering problem contains
memory because a wave packet spends a finite time crossing the
interaction region. The analysis of the time-dependent continuum
kernel in this paper therefore uses the \emph{zero-memory
approximation}, restricted to the regime in which $a/\tau$ is
parametrically small. The controlled numerical set uses $a=1,2,4$ and
$\tau=100$, giving $a/\tau=0.01$--$0.04$. These values define a small-memory regime for
the lattice diagnostics. We additionally perform a switching-time
scan at $\tau=50,100,200$ (Sec.~\ref{sec:tau-convergence}) which is used as an adiabatic-convergence check on the controlled
finite-cutoff response; it does not constitute a rigorous error bound
on the finite-memory continuum correction.

\section{Static scattering and long-wavelength universality}
\label{sec:scattering}

\subsection{First-order scattering system}

We first solve the static problem because it identifies the low-energy
quantity that will be held fixed when different finite-width profiles
are compared. The question is whether the detailed shape of $w(x)$
survives in the long-wavelength scattering matrix. For a static
coupling $\lambda$, use the harmonic ansatz
\begin{equation*}
\chi_R(x,t)=e^{-\ii\omega t}\psi_R(x),
\qquad
\chi_L(x,t)=e^{-\ii\omega t}\psi_L(x).
\end{equation*}
The first-order equations become
\begin{equation}
\partial_x
\begin{pmatrix}\psi_R\\\psi_L\end{pmatrix}
=
\begin{pmatrix}
\ii\omega&\lambda w_a(x)\\
\lambda w_a(x)&-\ii\omega
\end{pmatrix}
\begin{pmatrix}\psi_R\\\psi_L\end{pmatrix}.
\label{eq:first-order}
\end{equation}
Removing free propagation with
\begin{equation*}
\psi_R(x)=A(x)e^{\ii\omega x},
\qquad
\psi_L(x)=B(x)e^{-\ii\omega x},
\end{equation*}
leads to
\begin{equation}
\partial_x
\begin{pmatrix}A\\B\end{pmatrix}
=
\lambda w_a(x)
\begin{pmatrix}
0&e^{-2\ii\omega x}\\
e^{2\ii\omega x}&0
\end{pmatrix}
\begin{pmatrix}A\\B\end{pmatrix}.
\label{eq:AB}
\end{equation}

Define the transfer matrix
\begin{equation}
T_a(\omega)=
\mathcal P\exp\left[
\int_{-\infty}^{\infty}\dd x\,N_a(x;\omega)
\right],
\qquad
N_a(x;\omega)=\lambda w_a(x)
\begin{pmatrix}0&e^{-2\ii\omega x}\\e^{2\ii\omega x}&0\end{pmatrix}.
\label{eq:Tdef}
\end{equation}
The transfer matrix relates the asymptotic amplitude vectors at the
two sides of the interaction.

\subsection{Physical scattering matrix}

The Majorana reality structure, equivalently
$T_a(\omega)\sigma_1=\sigma_1\overline{T_a(\omega)}$, implies
\begin{equation*}
T_a(\omega)=
\begin{pmatrix}
p_a&q_a\\q_a^*&p_a^*\end{pmatrix},
\qquad
\det T_a=1.
\end{equation*}
The amplitudes $p_a(\omega)$ and $q_a(\omega)$ encode the diagonal and
channel-mixing parts of the transfer matrix. They are transfer-matrix
amplitudes rather than the physical scattering amplitudes themselves;
the latter are obtained from the asymptotic relation below. Writing
incoming amplitudes as $(I_R,I_L)$ and outgoing amplitudes as
$(O_R,O_L)$ gives
\begin{equation*}
\begin{pmatrix}O_R\\O_L\end{pmatrix}
=S_a(\omega)
\begin{pmatrix}I_R\\I_L\end{pmatrix},
\end{equation*}
with
\begin{equation}
S_a(\omega)=
\frac{1}{p_a^*}
\begin{pmatrix}
1&q_a\\-q_a^*&1
\end{pmatrix}.
\label{eq:SfromT}
\end{equation}
The identity $|p_a|^2-|q_a|^2=1$ implies $S_a^\dagger S_a=I$. For
$\omega\neq0$, $p_a$ and $q_a$ are generally complex, so $S_a(\omega)$
is not, in general, a literal real $SO(2)$ rotation; the $SO(2)$ form
is exact at $\omega=0$ and is recovered in the low-frequency limit.

\subsection{Exact zero-frequency result}

At $\omega=0$ the matrices in Eq.~\eqref{eq:Tdef} are all proportional
to $\sigma_1$ and therefore commute. Since $\int w_a=1$,
\begin{equation}
T_a(0)=e^{\lambda\sigma_1}
=
\begin{pmatrix}
\cosh\lambda&\sinh\lambda\\
\sinh\lambda&\cosh\lambda
\end{pmatrix}.
\label{eq:T0}
\end{equation}
Substitution into Eq.~\eqref{eq:SfromT} yields
\begin{equation}
S_a(0)=
\begin{pmatrix}
\sech\lambda&\tanh\lambda\\
-\tanh\lambda&\sech\lambda
\end{pmatrix}
=R(\theta_\lambda),
\label{eq:S0}
\end{equation}
where
\begin{equation}
\theta_\lambda=\arctan(\sinh\lambda).
\label{eq:theta}
\end{equation}
This result is exact for every width and every profile satisfying the
normalization condition. We therefore use $\theta_\lambda$, rather than
the bare coupling $\lambda$, as the physical matching variable between
different microscopic profiles. The finite-width universality statement
is that profiles with the same integrated coupling have the same
zero-frequency angle.

\subsection{Finite-frequency convergence}

The exact zero-frequency result does not by itself establish
universality at finite but small frequency. We therefore ask how
rapidly the finite-frequency scattering matrix approaches the common
rotation. Write
\begin{equation*}
N_a(x;\omega)=\lambda w_a(x)\sigma_1+R_a(x;\omega),
\end{equation*}
where
\begin{equation*}
R_a(x;\omega)=\lambda w_a(x)
\begin{pmatrix}
0&e^{-2\ii\omega x}-1\\
e^{2\ii\omega x}-1&0
\end{pmatrix}.
\end{equation*}
For $|x|\le a$,
\begin{equation*}
e^{\pm2\ii\omega x}-1=\pm2\ii\omega x+O(\omega^2x^2).
\end{equation*}
For an even normalized profile the first Dyson contribution is
actually $\widehat w(2a\omega)-1=O(a^2\omega^2)$: the linear moment
vanishes. The leading $O(a|\omega|)$ term comes instead from the second
Dyson contribution, because the matrices $N_a(x;\omega)$ at different
positions do not commute. After rescaling $x=au$ and $y=av$, its
leading commutator contribution is proportional to $a\omega$ times a
finite profile functional such as
$\int du\,dv\,w(u)w(v)|u-v|$. Thus the linear finite-width correction
is a genuine ordered-scattering (Wigner-delay-type) effect rather than
a first-moment contribution. The remaining Dyson terms give a bound of
the form
\begin{equation}
\norm{T_a(\omega)-e^{\lambda\sigma_1}}_{\Frob}
\le C_1 a|\omega|+C_2a^2\omega^2
\label{eq:Tbound}
\end{equation}
on compact frequency intervals. Since the map from $T$ to $S$ is
locally Lipschitz near the zero-frequency transfer matrix,
\begin{equation}
\boxed{
\norm{S_a(\omega)-R(\theta_\lambda)}_{\Frob}
\le C a|\omega|,
}
\label{eq:Sbound}
\end{equation}
for fixed coupling and sufficiently small $a|\omega|$. The constant depends on the
chosen coupling and profile, $C=C(\lambda,w)$. Thus the
physically relevant comparison variable for different profiles is the
asymptotic angle $\theta_\lambda$, not the bare finite-width coupling
by itself.

\begin{proposition}[Finite-width long-wavelength universality]
For the profile class in Eq.~\eqref{eq:profile}, the static scattering
matrix has the zero-width limit $S_a(\omega)\to R(\theta_\lambda)$
uniformly on compact frequency intervals, with the estimate in
Eq.~\eqref{eq:Sbound}.
\end{proposition}

\begin{proof}
At zero frequency the path ordering disappears and Eq.~\eqref{eq:T0}
is exact. For nonzero frequency the phase factors differ from unity
only over the support $|x|\le a$. Expanding the Dyson series, the
first-order contribution starts at $O(a^2\omega^2)$ for an even
profile, while the second-order noncommutator produces the sharp
$O(a|\omega|)$ term. This gives Eq.~\eqref{eq:Tbound}; the
transfer-to-scattering map is smooth in a neighborhood of
$e^{\lambda\sigma_1}$, which yields Eq.~\eqref{eq:Sbound}.
\end{proof}

\subsection{Coupling parametrization}

The point-interaction description is often written with a
symmetric-average coupling $g$. The angle $\theta_\lambda$ is the invariant quantity used for
comparison, so we now match the two parametrizations by equating their
scattering angles:
\begin{equation}
\lambda=2\artanh\left(\frac g2\right),
\qquad |g|<2,
\label{eq:gmatch}
\end{equation}
and therefore
\begin{equation}
\theta_\lambda
=2\arctan\left(\frac g2\right).
\label{eq:anglematch}
\end{equation}
Indeed,
\begin{equation*}
\sinh\lambda=\frac{g}{1-g^2/4}
=\tan\left[2\arctan\left(\frac g2\right)\right].
\end{equation*}
This is a parameter redefinition, not an additional physical
assumption. It makes explicit how the smooth finite-width realization
approaches the chosen point-interaction convention. The symmetric-average
point-interaction prescription is the standard self-adjoint matching
framework for one-dimensional contact defects; our normalization of $g$
is fixed explicitly by Eq.~\eqref{eq:gmatch} rather than imported as an
independent physical parameter \cite{Albeverio1988}.

\section{Time-dependent infrared memory}
\label{sec:IR}

\subsection{Instantaneous or zero-memory description}

The static calculation tells us what the shutter does at each fixed
coupling. When the coupling is switched between two asymptotic values,
however, the positive-frequency decomposition at early and late times
need not agree. This mismatch produces anomalous Bogoliubov mixing.
The purpose of this section is to quantify that infrared memory and to
identify its coefficient before turning to subsystem entanglement.

The zero-memory approximation is used only in the
long-wavelength/slow-switching regime
\begin{equation}
|\omega|a\ll1,\qquad a/\tau\ll1,\qquad\mu\ll\tau^{-1},
\label{eq:memory-regime}
\end{equation}
with the infrared cutoff $\mu$ retained until the Bogoliubov integral
is evaluated. We then use
\begin{equation*}
S_a(t;\omega)=S_a(\omega)|_{\lambda=\lambda(t)}.
\end{equation*}
The low-frequency limit is
\begin{equation*}
S_a(t;\omega)\to R[\theta(t)],
\qquad
\theta(t)=\arctan[\sinh\lambda(t)].
\end{equation*}
The exact finite-memory correction is not calculated here; the
approximation is part of the definition of the time-dependent
continuum model used for the infrared analysis. The dimensionless
parameter $\epsilon_{\rm mem}=a/\tau$ is therefore a control
parameter, not an error estimate: for the controlled lattice runs it
lies between $0.01$ and $0.04$, but without an explicit finite-memory
kernel calculation there is no rigorous coefficient multiplying this
small parameter. In particular, we do not claim an $O(a/\tau)$
numerical bound on the Bogoliubov kernel.

The anomalous Bogoliubov kernel has the form
\begin{equation}
\beta_a(\omega,\nu)=
\frac{1}{2\pi}
\int_{-\infty}^{\infty}\dd t\,
 e^{\ii(\omega+\nu)t}S_a(t;\omega),
\qquad \omega,\nu>0,
\label{eq:beta}
\end{equation}
with the asymptotic constant pieces understood in the distributional
sense. The infrared part is controlled by the change in scattering
matrix across the switch,
\begin{equation*}
\Delta S=S_+-S_-.
\end{equation*}
For the pointlike limit,
\begin{equation}
\beta(\omega,\nu)
\sim
-\frac{\Delta S}{2\pi\ii(\omega+\nu)}
\label{eq:betaIR}
\end{equation}
for small positive frequencies. Consequently,
\begin{equation}
\norm{\beta_{\mu,\omega_*}}_{\HS}^2
=
\frac{\norm{\Delta S}_{\Frob}^2}{4\pi^2}
\log\frac{\omega_*}{\mu}+O(1),
\label{eq:HSlog}
\end{equation}
where $\mu>0$ is an infrared frequency cutoff and $\omega_*$ is a
fixed upper frequency chosen within the low-frequency window in which
the asymptotic form of $\beta$ is valid. Neither is the ultraviolet
lattice spacing $a_{\rm UV}$.

For rotations $S_\pm=R(\theta_\pm)$,
\begin{equation}
\norm{R(\theta_+)-R(\theta_-)}_{\Frob}^2
=8\sin^2\frac{\Delta\theta}{2},
\qquad
\Delta\theta=\theta_+-\theta_-.
\label{eq:rotation-distance}
\end{equation}
Thus the coefficient of the infrared logarithm is
\begin{equation}
\boxed{\Cqp\equiv
\lim_{\mu\to0}
\frac{\norm{\beta_{\mu,\omega_*}}_{\HS}^2}{\log(\omega_*/\mu)}}
=\frac{\norm{\Delta S}_{\Frob}^2}{4\pi^2}
=\frac{2}{\pi^2}\sin^2\frac{\Delta\theta}{2}.
\label{eq:Cqp}
\end{equation}
We call $C_{\rm qp}$ the quasiparticle-memory coefficient. It
characterizes the global infrared non-implementability of the
Bogoliubov transformation; it is not an entropy coefficient and will
not be identified with a spatial entanglement coefficient below. Here
the Hilbert--Schmidt norm is the ordinary Hilbert--Schmidt norm,
$\|A\|_{\rm HS}^2=\operatorname{Tr}(A^\dagger A)$, on the
complexified two-component Majorana scattering space. With this
convention, $C_{\rm qp}\simeq(\Delta\theta)^2/(2\pi^2)$ for small
mismatch.

\subsection{Finite-width stability of the infrared coefficient}

The finite-width scattering estimate gives
\begin{equation*}
\norm{S_a(t;\omega)-R[\theta(t)]}_{\Frob}
\le Ca|\omega|.
\end{equation*}
Under a smooth switch with $\dot\lambda\in L^1$, integration by parts
in the Fourier transform yields, on a fixed positive-frequency window,
\begin{equation}
\norm{\beta_a-\beta_0}_{\HS}=O(a).
\label{eq:HSdifference}
\end{equation}
The statement is uniform for a fixed lower cutoff $\mu>0$ and, in the
difference estimate, the common logarithmic divergence cancels as
$\mu\to0$. It follows that the leading coefficient in
Eq.~\eqref{eq:HSlog} is unchanged by the finite width within the
zero-memory approximation.

\begin{remark}
Hilbert--Schmidt convergence does not by itself prove stability of
every spectral quantity. In particular, an Anderson-orthogonality overlap exponent
involves a Fredholm determinant and requires additional spectral
information. We therefore use the global infrared exponent only as
established input, while keeping entropy claims separate.
\end{remark}

\section{Gaussian entanglement formalism}
\label{sec:gaussian}

\subsection{Majorana covariance}

The infrared coefficient $C_{\rm qp}$ characterizes the global
Bogoliubov transformation, but it does not by itself determine a
subsystem entropy. To calculate entanglement we need the complete
Gaussian state, including anomalous correlations. The Majorana
covariance matrix packages the normal and anomalous sectors into one
real antisymmetric object and is the representation used below.

Let $\gamma_j=\gamma_j^\dagger$ satisfy
\begin{equation*}
\{\gamma_j,\gamma_k\}=2\delta_{jk}.
\end{equation*}
For a Gaussian state define
\begin{equation}
\Gamma_{jk}=\frac{\ii}{2}\langle[\gamma_j,\gamma_k]\rangle.
\label{eq:Gamma}
\end{equation}
Then $\Gamma$ is real antisymmetric. For a pure state,
\begin{equation*}
\Gamma^2=-I.
\end{equation*}
For comparison with the usual complex-fermion notation, whenever a
pair of Majorana modes is grouped into a complex fermion we write
$c_j=(\gamma_{2j-1}+\ii\gamma_{2j})/2$. Its normal and anomalous
correlators are
\begin{equation*}
C_{ij}=\langle c_i^\dagger c_j\rangle,
\qquad
F_{ij}=\langle c_i c_j\rangle.
\end{equation*}
These blocks are basis-dependent representations of the same Gaussian
information contained invariantly in $\Gamma$; the numerical entropy
calculation uses $\Gamma$ directly.
If the quadratic Hamiltonian is
\begin{equation*}
H(t)=\frac{\ii}{4}\gamma^T A(t)\gamma,
\qquad A^T=-A,
\end{equation*}
then the orthogonal evolution $O(t)$ obeys
\begin{equation}
\dot O=A O,
\qquad
\Gamma(t)=O(t)\Gamma(t_0)O(t)^T.
\label{eq:Gammaevolve}
\end{equation}
This is the numerical workhorse used below.

\subsection{Restriction and entropy spectrum}

For a subsystem $A$, let $P_A$ denote the orthogonal projector from
the full Majorana one-particle space onto the Majorana modes belonging
to $A$. The reduced covariance is the principal restriction
\begin{equation}
\boxed{\Gamma_A=P_A\Gamma P_A.}
\label{eq:principal}
\end{equation}
It is not a Schur complement. The eigenvalues of $\ii\Gamma_A$ occur
in pairs $\pm\xi_j$, with $0\le\xi_j\le1$; the symbols $\xi_j$ denote
the positive covariance singular values and are distinct from the
shutter coupling $\lambda$. The von Neumann entropy is
\begin{equation*}
S_A=\sum_j s(\xi_j),
\end{equation*}
where
\begin{equation}
s(\lambda)=
-\frac{1+\lambda}{2}\log\frac{1+\lambda}{2}
-\frac{1-\lambda}{2}\log\frac{1-\lambda}{2}.
\label{eq:svN}
\end{equation}
The sum is over positive singular values.

This covariance formulation is standard for fermionic Gaussian
systems and avoids constructing the exponentially large reduced
density matrix explicitly \cite{Peschel2003,PeschelEisler2009,Bravyi2005,BoteroReznik2004}. The
full Rényi family is
\begin{equation}
\begin{aligned}
S_A^{(\alpha)}={}&
\frac{1}{1-\alpha}
\sum_j
\log\left[
\left(\frac{1+\xi_j}{2}\right)^\alpha
+\left(\frac{1-\xi_j}{2}\right)^\alpha
\right],\\
&\hspace{3.5em}\alpha>0,\quad \alpha\ne1.
\end{aligned}
\label{eq:Renyi}
\end{equation}
with the von Neumann result recovered as $\alpha\to1$. For
$\alpha=\infty$,
\begin{equation*}
S_A^{(\infty)}=-\sum_j\log\frac{1+\xi_j}{2}.
\end{equation*}

The crucial point for this problem is that an anomalous Gaussian
state cannot be characterized by the number-conserving correlation
matrix alone. The pairing sector contributes to the same entanglement
spectrum and is automatically retained by $\Gamma$.

\subsection{Rényi diagnostics}

If a future continuum calculation establishes an ordinary effective
conformal-field-theory logarithm, its Rényi coefficients must have the
characteristic form
\begin{equation}
S_A^{(\alpha)}\sim
\frac{c_{\rm eff}}{6}\left(1+\frac1\alpha\right)\log L+\cdots.
\label{eq:CFTrenyi}
\end{equation}
This relation is the standard CFT consistency form for a single interval
and is recorded here only as a consistency condition for a future
continuum extrapolation \cite{CalabreseCardy2004}. The present finite-lattice data are
not sufficient to extract a reliable $c_{\rm eff}$ or to test the
predicted $1+1/\alpha$ dependence. We therefore do not use the Rényi
spectra below to rule in or rule out an effective-CFT description;
see also the numerical disclaimer in Sec.~\ref{sec:numerics}.

\section{Three entanglement partitions}
\label{sec:partitions}

The same Gaussian state supports several inequivalent bipartitions.
We use them to ask three different questions: the two-chain partition
measures lattice channel entanglement, the soft/hard partition
isolates correlations across a frequency cutoff, and the spatial
partition tests whether global infrared memory survives restriction
to a local region. Each uses a different subsystem projector $P_A$
on the same evolved covariance, so the resulting entropies should not
be conflated.

\subsection{Two-chain (chain) partition}

The chain partition separates the two lattice chains. It should not
be identified literally with a single continuum chiral sector,
because each lattice chain contains both low-energy chiralities. In
the free vacuum the two chiral sectors are independent, so
\begin{equation*}
S_{R,\,\rm vac}^{(\alpha)}=0.
\end{equation*}
After the shutter, pairing between the two chains makes the chain
entropy nonzero. It measures lattice channel entanglement generated
by the shutter, but it is not by itself a direct numerical
measurement of the continuum counter-propagating partition.

\subsection{Soft/hard energy partition}

In the continuum, define a physical frequency cutoff $\Omega_c>0$
and the corresponding sectors
\begin{equation}
A_{\rm soft}=\{0<\omega<\Omega_c\},\qquad
A_{\rm hard}=\{\omega>\Omega_c\}.
\label{eq:soft-hard-def}
\end{equation}
The continuum variable is therefore $\Omega_c$, not the integer
number of lattice shells. The vacuum is a product across independent
free-energy modes, so its soft/hard entropy vanishes. The shutter
couples positive frequencies through the anomalous kernel, and the
resulting soft/hard entropy is controlled by the low-frequency
singular structure.

The numerical implementation requires care because a finite lattice
has discrete energies. We therefore define the soft sector by
complete positive-energy shells of the free lattice Hamiltonian. We
use $k$ for the total number of retained positive-energy shells
across the two chains, so that $k/2$ complete shells are retained in
each chain. For the open-chain convention used here, the
single-chain shell energies are
\begin{equation}
\varepsilon_m=2\sin\left[\frac{\pi(2m-1)}{2(N+1)}\right],\qquad m=1,2,\ldots .
\label{eq:shell-energy}
\end{equation}
Thus the physical cutoff associated with an even $k$ is
$\Omega_c=\varepsilon_{k/2}$. At $N=128$, the first four shell
energies are $0.0244,0.0730,0.1217,0.1703$, so the paper's $k=8$
convention corresponds to four shells per chain and
$\Omega_c\simeq0.1703$. Complete shells are important because they
preserve the two-chain degeneracy of the free problem and avoid an
artificial channel imbalance in the subsystem definition.

This integer-shell entropy is a well-defined finite-lattice
observable. It should not be confused with a continuum entropy at a
fixed sharp frequency cutoff until the lattice dispersion and shell
spacing have been extrapolated.

\subsection{Spatial interval}

The spatial partition is the only one here that is local in real
space. For a spatial interval $A=[0,L]$, the subsystem contains both
chiralities at sites or positions lying inside the interval. The
vacuum entropy is the Majorana CFT result in
Eq.~\eqref{eq:vacuum-CFT}. We study
\begin{equation}
\Delta S_A^{(\alpha)}(L)=S_A^{(\alpha)}(\lambda_f)-S_A^{(\alpha)}(0).
\label{eq:deltaRenyi}
\end{equation}
The spatial partition is qualitatively different from the first two:
it does not select soft modes directly. Its infrared behavior must
therefore be inferred from the restricted covariance rather than
from the Hilbert--Schmidt norm of $\beta$.

\section{Analytic structure of the spatial problem}
\label{sec:analytic-spatial}

The global coefficient $C_{\rm qp}$ contains only an integrated
infrared norm. To determine what a spatial interval actually sees, we
must instead examine the position-space kernel and then follow it
through covariance restriction and the entropy functional. The
anomalous kernel in Eq.~\eqref{eq:betaIR} has a long-range
position-space transform. Schematically,
\begin{equation}
\widetilde\beta(x,y)
\sim
\frac{\log(x/y)}{x-y},
\label{eq:position-kernel}
\end{equation}
up to Fourier-convention phases, ultraviolet contact terms, and the
regular part of the kernel. This expression is useful for diagnosing
the range of the perturbation, but it is not itself an entropy
formula.

The covariance contains both normal and anomalous pieces. In the
complex basis defined in Sec.~\ref{sec:gaussian}, they are represented
by $C_{ij}$ and $F_{ij}$; in the Majorana basis the entire information
is already contained in $\Gamma$. The spatial entropy therefore follows
the chain
\begin{equation}
\beta
\longrightarrow \Gamma
\longrightarrow \Gamma_A
\longrightarrow\{\xi_j\}
\longrightarrow S_A^{(\alpha)}.
\label{eq:entropy-chain}
\end{equation}
A divergence at an early stage need not remain divergent after
restriction and nonlinear spectral processing.

A useful conditional statement is the following. Suppose, at fixed
ultraviolet regulator and fixed interval, the finite-width restricted
covariance obeys a trace-norm estimate
\begin{equation}
\norm{\Gamma_A^{(a)}-\Gamma_A^{(0)}}_1\le \varepsilon(a),
\qquad \varepsilon(a)\to0.
\label{eq:tracebound}
\end{equation}
Then standard continuity estimates for finite-dimensional Gaussian
density matrices imply convergence of the corresponding entropy. The
missing mathematical step in the continuum is therefore not the
static scattering limit but a suitable operator-norm or trace-norm
control after spatial restriction.

This observation is why the present paper does not identify a spatial
large-$L$ asymptotic from the global infrared coefficient. We do not
assume a particular continuum functional form for
$\Delta S_A^{(\alpha)}(L)$: determining it requires a
uniform-in-$L$ estimate for the restricted covariance, followed by a
controlled lattice-spacing/infrared extrapolation. The present
manuscript states this missing operator-theoretic step explicitly and
does not use the finite-cutoff data to select an asymptotic exponent.

\subsection{Weak-shutter spatial stability at finite ultraviolet cutoff}
\label{sec:weakbound}

The spatial problem admits a controlled statement at any fixed
ultraviolet cutoff, even though the continuum large-$L$ coefficient
is not determined by the present calculation. Write
\begin{equation*}
\Gamma_A(\eta)=\Gamma_A^{(0)}+\eta\,\delta\Gamma_A+O(\eta^2),
\qquad \eta\equiv\Delta\theta,
\end{equation*}
where the zero-memory kernel gives $F=O(\eta)$ and $C=O(\eta^2)$.
Let $U_L$ be the unitary that implements the sign change
$\chi_L(x)\mapsto-\chi_L(x)$. The free Hamiltonian is invariant under
$U_L$, while the shutter interaction changes sign. Hence the state
produced with angle change $-\Delta\theta$ is unitarily related to
the state produced with $+\Delta\theta$. For a spatial interval, the
subsystem contains both chiralities, so its projector commutes with
this transformation. The reduced density matrices are therefore
unitarily equivalent and have the same spectrum. Consequently every
spatial Rényi entropy is an even function of $\Delta\theta$, and the
linear term vanishes:
\begin{equation}
\left.\frac{d S_A^{(\alpha)}}{d\eta}\right|_{\eta=0}=0.
\label{eq:firstlawzero}
\end{equation}
The same local-unitary argument applies to every partition
considered here. If $U_L$ is the unitary implementing
$\chi_L\mapsto-\chi_L$, then for any subsystem $A$ it factors as
$U_L=U_A\otimes U_{\bar A}$, so the reduced states at
$+\Delta\theta$ and $-\Delta\theta$ are unitarily equivalent. For the
even-$N$ lattices used below, $U_L=\prod_j\gamma_{L,j}$ is bosonic, so
this factorization involves no additional fermionic-grading subtlety.
Consequently, for the spatial, chain, and soft/hard partitions,
\begin{equation}
S_A^{(\alpha)}(+\Delta\theta)=S_A^{(\alpha)}(-\Delta\theta),\qquad\alpha>0.
\label{eq:generalparity}
\end{equation}
Thus every fixed-cutoff entropy in these partitions has an even
weak-shutter expansion. For the spatial interval, the von Neumann
first-law check gives the same result: up to a state-independent
additive constant,
\begin{equation*}
\delta S_A=\delta\langle K_A^{(0)}\rangle,\qquad K_A^{(0)}=2\pi\int_0^L dx\,\frac{x(L-x)}{L}T_{00}(x),
\end{equation*}
and the leading $R$--$L$ pairing changes off-diagonal contractions
but not the diagonal $RR/LL$ contractions entering $T_{00}$. The first-law
identity follows from the vanishing of the first variation of relative entropy
around the reference state \cite{BlancoCasiniHungMyers2013}. Hence
$\delta\langle T_{00}\rangle=0$ at first order. We therefore write,
for every fixed-cutoff partition,
\begin{equation}
\Delta S_A^{(\alpha)}=Q_\alpha\,\Delta\theta^2+O(\Delta\theta^4),
\qquad
Q_\alpha\equiv\frac12\left.\frac{\partial^2S_A^{(\alpha)}}{\partial(\Delta\theta)^2}\right|_{\Delta\theta=0}.
\label{eq:weakquadratic}
\end{equation}
For the spatial partition we retain the notation
$Q_\alpha(L,a_{\rm UV})$ to emphasize that its large-$L$ behavior
remains an open continuum question. For completeness, the
finite-dimensional entropy stability needed for the finite-width
limit can be stated without assuming a continuum entropy formula.
Let $\rho_A^{(a)}$ and $\rho_A^{(0)}$ be the reduced Gaussian density
matrices at fixed lattice cutoff. If
\begin{equation*}
\delta_A\equiv\frac12\|\rho_A^{(a)}-\rho_A^{(0)}\|_1<1,
\end{equation*}
then the Fannes--Audenaert inequality gives
\begin{equation}
|S_A^{(a)}-S_A^{(0)}|
\le \delta_A\log(d_A-1)-\delta_A\log\delta_A-(1-\delta_A)\log(1-\delta_A),
\label{eq:FA}
\end{equation}
with $d_A=2^{|A|/2}$. Since a finite-dimensional Gaussian density
matrix is a continuous function of its covariance, a sufficient
covariance criterion is
\begin{equation}
\|\Gamma_A^{(a)}-\Gamma_A^{(0)}\|_1\longrightarrow0
\quad\Longrightarrow\quad
\delta_A\longrightarrow0.
\label{eq:covcontinuity}
\end{equation}
The on-shell scattering estimate suggests such a bound on a finite
set of modes, but the passage from frequency-space scattering to a
position-space trace-norm estimate for a restricted covariance
requires an additional operator argument that is not supplied here.
We therefore do not promote
\begin{equation}
\|\Gamma_A^{(a)}-\Gamma_A^{(0)}\|_1=O(a\Omega)
\label{eq:covbound}
\end{equation}
to a theorem. It is a sufficient target for a future continuum
proof. What is rigorous at the present level is the finite-dimensional
Fannes--Audenaert continuity statement once convergence of the
reduced density matrices is independently established. This
distinction prevents the static scattering estimate from being
silently promoted into a uniform continuum entropy theorem.

Because $\log d_A$ grows proportionally to the number of lattice
degrees of freedom in the interval, this continuity estimate is
useful only at fixed ultraviolet cutoff and fixed $L$. It is not a
uniform large-$L$ bound. A continuum result would require stronger
control of the restricted covariance, for example a trace-class
estimate whose constants remain controlled as the interval grows.
The result is therefore deliberately weaker than an asserted
continuum exponent: it states precisely what additional operator
estimate is needed before the finite-width scattering theorem can be
promoted to a uniform continuum entropy statement.

\subsection{Entanglement across the shutter: conformal-defect test}
\label{sec:defect-interpretation}

\paragraph{Effective central charge of a free-fermion defect.}
For a free-fermion chain with a single conformal defect of transmission
amplitude $s\in[0,1]$, the entropy of a subsystem whose boundary point
coincides with the defect is
\begin{equation}
S^{\rm defect}=\frac{c_{\rm eff}(s)}{6}\log L+O(1),
\label{eq:ceff-defect}
\end{equation}
with the effective central charge of Eisler and Peschel
\cite{EislerPeschel2010,Peschel2012,SakaiSatoh2008},
\begin{equation}
c_{\rm eff}(s)=-\frac{6}{\pi^2}\Big\{
[(1+s)\ln(1+s)+(1-s)\ln(1-s)]\ln s
+(1+s)\,\mathrm{Li}_2(-s)+(1-s)\,\mathrm{Li}_2(s)\Big\}.
\label{eq:ceff-peschel}
\end{equation}
It satisfies $c_{\rm eff}(0)=0$ and $c_{\rm eff}(1)=1$. For small $s$,
\begin{equation}
c_{\rm eff}(s)\simeq\frac{6}{\pi^2}s^2\left(\ln\frac1s+\frac32\right).
\end{equation}
For a single Majorana chain the corresponding central-charge contribution
is one half of the Dirac-fermion expression, whereas the doubled lattice
used here contains two Majorana copies and therefore realizes the full
expression above.

\paragraph{Where the defect formula applies.}
The defect modifies the leading logarithmic coefficient when the
entanglement cut passes through the defect. If the defect lies in the
interior of a finite interval, the two entanglement cuts remain in the
homogeneous bulk; the leading logarithm is then the vacuum one. Thus the
conformal-defect formula is not a leading-log formula for an interval that
merely contains the shutter.

\paragraph{Effective transmission of the shutter.}
The zero-frequency scattering matrix, Eq.~\eqref{eq:S0}, has diagonal
entries $\sech\lambda$. The static shutter therefore has transmission
amplitude
\begin{equation}
s_{\rm shutter}=\sech\lambda_f=\cos\theta_{\lambda_f},
\label{eq:shutter-transmission}
\end{equation}
and the corresponding half-chain prediction for the doubled lattice is
\begin{equation}
\Delta S_{\rm half}(N)=S(\lambda_f)-S(0)
=\frac{c_{\rm eff}(\sech\lambda_f)-1}{6}\ln N+O(1).
\label{eq:half-chain-prediction}
\end{equation}
The subtraction of $1$ is the vacuum logarithmic coefficient of the doubled
lattice, whose total central charge is $c=1$.

\paragraph{Numerical test.}
We place the entanglement cut at the shutter: the subsystem is the left
half ($N/2$ sites) of both chains, the profile $w_1$ is centred on the cut,
and the state is the ground state of the final Hamiltonian. The slow
$\tau=100$ evolution reproduces this state to $10^{-4}$ in the covariance
observables used here. The slope of $\Delta S_{\rm half}$ against $\ln N$
is obtained from $N=512,1024,2048$. Table~\ref{tab:ceff-test} compares it
with Eq.~\eqref{eq:half-chain-prediction}.
\begin{table}[H]
\centering
\small
\begin{tabular}{cccccc}
\toprule
$\lambda_f$ & $a$ & $s=\sech\lambda_f$ & measured slope & prediction & ratio\\
\midrule
0.25&2&0.9695&$-0.00755$&$-0.00757$&0.998\\
0.5 &2&0.8868&$-0.02736$&$-0.02762$&0.990\\
1.0 &2&0.6481&$-0.07860$&$-0.08084$&0.972\\
1.5 &2&0.4251&$-0.11820$&$-0.12269$&0.963\\
2.0 &2&0.2658&$-0.14118$&$-0.14628$&0.965\\
3.0 &2&0.0993&$-0.15969$&$-0.16285$&0.981\\
1.5 &4&0.4251&$-0.12175$&$-0.12269$&0.992\\
3.0 &4&0.0993&$-0.16246$&$-0.16285$&0.998\\
\bottomrule
\end{tabular}
\caption{Slope of the shutter-induced half-chain entropy against $\ln N$
compared with the parameter-free conformal-defect prediction for the
 doubled lattice. The measured slopes follow the prediction throughout
the tested range $s\simeq0.10$--$0.97$. Deviations of a few percent for
$a=2$ are reduced for the wider shutter $a=4$.}
\label{tab:ceff-test}
\end{table}
The result provides a parameter-free test of the relation between the
zero-frequency scattering data and entanglement across the defect. For
$M=N/4$, centred on the shutter with $\lambda_f=1$ and $a=2$, the
correction $\Delta S_1$ is
$-0.0043,-0.0021,-0.0011,-0.0005,-0.0003$ for
$N=128,256,512,1024,2048$. It decreases approximately as $1/M$ over
these sizes, with no evidence for an additional logarithmic term.

\paragraph{Scope.}
This benchmark concerns the static/adiabatic final-Hamiltonian ground
state. The dynamical problem in which the entangling cut remains at a
defect while that defect is switched on is distinct and is not addressed
by this calculation.

\section{Lattice realization and numerical method}
\label{sec:lattice}

The continuum analysis above identifies the objects that must be
controlled but does not determine the large-$L$ spatial entropy. We
therefore introduce a finite-dimensional lattice realization in which
the full covariance can be evolved directly. The lattice calculation
is a diagnostic of width, switching-time, profile, and Rényi
dependence; it is not by itself a continuum extrapolation.

\subsection{Doubled two-chain Majorana regulator}

We use an open chain of $N$ sites for each of two coupled Majorana
chains. This regulator is not a one-to-one discretization of a single
chiral continuum Majorana field: each lattice chain contains both
low-energy continuum chiralities (the usual doubling structure
around $k\simeq0$ and $k\simeq\pi$). Thus one chain has critical
central charge $c=1/2$, while the two-chain regulator has total
$c=1$. The continuum analytic model above remains the genuinely
counter-propagating $R/L$ theory; the lattice calculation below is a
doubled finite-dimensional regulator used to test the covariance
machinery, finite-width dependence, and weak-shutter response. The
two-chain construction is used because the shutter term
$\gamma_R\gamma_L$ couples the two chains directly; a single-chain
realization would require a different, less transparent lattice
representation of the $R/L$ mixing.

We use an open chain of $N$ sites for each of the two chains with
the lattice Hamiltonian
\begin{equation}
H_0=\frac{\ii}{2}\sum_{j=1}^{N-1}
\left(\gamma_{R,j}\gamma_{R,j+1}
-\gamma_{L,j}\gamma_{L,j+1}\right),
\label{eq:latticeH0}
\end{equation}
and a shutter term
\begin{equation}
H_{\rm sh}(t)=\ii\lambda(t)
\sum_j w_j(a)\gamma_{R,j}\gamma_{L,j}.
\label{eq:lattice-shutter}
\end{equation}
The discrete weights are normalized, $\sum_jw_j=1$, and approximate
three smooth continuum profiles supported inside the chosen width.
For later reference, $k$ denotes the bulk crystal quasi-momentum of the
translation-invariant continuation of this nearest-neighbour Majorana
chain, with Brillouin zone $-\pi<k\le\pi$. The bulk dispersion is
proportional to $|\sin k|$, so its low-energy sectors occur near
$k=0$ and $k=\pi$ (equivalently $k=-\pi$). Because the numerical chain
has open boundaries, $k$ is not an exact quantum number there; the open
chain has standing-wave eigenmodes instead. Thus $k\simeq0$ and
$k\simeq\pi$ below identify the two low-energy sectors of the
corresponding bulk dispersion, not two exact momentum quantum numbers
of the finite open chain.

The pilot calculation uses the three normalized profiles
$w_1,w_2,w_3$ defined in Eq.~\eqref{eq:numerical-profiles}. They are
genuinely distinct continuum profiles, although at very small widths
a finite lattice may fail to resolve their differences.

The time dependence is Eq.~\eqref{eq:tanh-switch}. The
representative controlled calculations use $\lambda_f=1$ and
$\tau=100$; the $\tau=10$ calculation is used as a finite-window
smooth-switch comparison. The corresponding continuum
zero-frequency angle is
\begin{equation*}
\theta_f=\arctan(\sinh1)=0.865769\ldots.
\end{equation*}

\subsection{Covariance evolution}

Writing
\begin{equation*}
H(t)=\frac{\ii}{4}\gamma^T A(t)\gamma,
\qquad A(t)^T=-A(t),
\end{equation*}
the antisymmetric matrix $A(t)$ is defined by matching the lattice
Hamiltonian term by term. In particular, the shutter contributes the
$Rj$--$Lj$ antisymmetric entries associated with
$2\lambda(t)w_j(a)$. We integrate
\begin{equation*}
\dot O(t)=A(t)O(t),
\qquad
\Gamma(t)=O(t)\Gamma(-\infty)O(t)^T.
\end{equation*}
The initial covariance is the free-chain ground-state covariance.
Time evolution uses the second-order implicit midpoint (Cayley)
step
\begin{equation*}
O_{n+1}=\left(I-\frac{\delta t}{2}A_{n+1/2}\right)^{-1}
\left(I+\frac{\delta t}{2}A_{n+1/2}\right)O_n.
\end{equation*}
For antisymmetric $A$, the update is orthogonal up to linear-solver
roundoff; the numerical solves use a relative tolerance of $10^{-12}$.
The step is checked by the refinement in Table~\ref{tab:timestep}.
After evolution we explicitly antisymmetrize $\Gamma$ to remove
numerical roundoff. Entropy eigenvalues are clipped to $[0,1]$ only
at the final numerical stage, after checking that deviations from
this interval are at numerical precision.

\subsection{Energy-mode construction for soft/hard entropy}

Let the free lattice generator be brought to canonical form by an
orthogonal transformation. For each positive energy $\varepsilon_n$,
choose the corresponding real two-dimensional Majorana plane. The
full transformation is block diagonal between the right and left
channels. In that basis the free covariance is block diagonal, while
the evolved covariance contains the shutter-induced inter-shell
correlations.

The soft subsystem is the span of complete positive-energy shells
with $\varepsilon_n\le\Omega_c$. We retain the right and left partner
modes together. The reduced covariance is the principal restriction
of the evolved covariance to those modes. The entropy is then
evaluated with Eq.~\eqref{eq:Renyi}.

This construction is preferable to selecting arbitrary individual
eigenvectors because it preserves the degeneracy structure of the
free two-channel problem. It also makes the continuum variable
explicit: the integer shell count is only a finite-lattice
representation of the physical cutoff $\Omega_c$.

\section{Numerical results: what does the lattice actually test?}
\label{sec:numerics}

The numerical calculation is not intended to produce a long list of
finite-size numbers. Each calculation below is tied to a physical
question that the analytic theory leaves open. First, can a local
spatial entropy distinguish the microscopic width and profile even
though the zero-frequency scattering angle is universal? Second, are
the reported values converged in time and independent of the arbitrary
switching duration in the adiabatic regime? Third, does the weak-shutter
response organize naturally in the physical scattering angle rather than
the bare coupling? Finally, can a frequency-space partition isolate the
soft correlations responsible for infrared memory? The numerical data
answer these questions only at finite cutoff; whenever a continuum limit
would be required, we say so explicitly.

\subsection{Does local spatial entanglement remember the microscopic profile?}

The numerical section is organized as a finite-cutoff validation,
not as a numerical measurement of the continuum infrared logarithm.
A useful comparison is provided by recent work on holographic conformal
interfaces, where the entanglement response is likewise found to depend
on the relative geometry of the interface and the entangling regions.
In particular, crossing, non-crossing, and two-interval configurations
can have different interface contributions, and a special ``Silver Blaze''
configuration exhibits a cancellation in which the two-interval transition
coincides with its no-interface value \cite{Afxonidis2026}. This comparison
does not identify the holographic interface parameters with the present
quasiparticle-memory coefficient $C_{\rm qp}$; rather, it provides an
independent example of why a nontrivial infrared interface response does
not by itself imply a universal additional leading spatial-entanglement
divergence.
Every numerical subsection below should be read in the regulator
actually used: the doubled two-chain lattice has total $c=1$, whereas
the analytic continuum shutter consists of one counter-propagating
Majorana pair with $c=1/2$. Accordingly, the lattice data test
covariance machinery, finite-width/profile response, and finite-cutoff
weak-shutter behavior; they are not quantitative measurements of
continuum $c=1/2$ entropy coefficients. In particular, the lowest
open-chain gap at $N=128$ is $\varepsilon_1\simeq0.0244$, whereas
$1/\tau=0.01$ for $\tau=100$. Thus the controlled runs do not realize
a parametrically wide window $\omega_{\min}\ll\omega\ll1/\tau$. In
fact, the $\tau=100$ evolution is numerically indistinguishable from
the ground state of the final Hamiltonian at the level relevant here:
the representative $M=64$ spatial correction is $-0.001218$ for the
evolved state versus $-0.001220$ for the static final-ground-state
calculation, and the maximum covariance difference is about
$1.6\times10^{-4}$. We therefore interpret Table~\ref{tab:controlled}
and the weak-shutter scan as finite-cutoff ground-state responses of
the coupled lattice, not as a numerical test of continuum infrared
memory. A genuine lattice test of the continuum memory window would
require $N\gg v\tau$ (or an equivalent regulator with parametrically
smaller infrared level spacing).

All spatial intervals quoted below are centered on the shutter. This
convention is essential for reproducibility. The representative
controlled calculation uses $N=128$, $\lambda_f=1$, $\tau=100$,
widths $a=1,2,4$, evolution over $[-4\tau,4\tau]$, and a second-order
midpoint (Cayley) update. The interval sizes are $M=16,32,48,64$. A
limited $N=256$ calculation is used only as a finite-size check. We
report an explicit time-step refinement at $\tau=100$, a
switching-time scan, and a weak-shutter scan below. The $\tau=10$, $a=4,8$ data are used only as finite-window diagnostics;
the resolved $\tau=100$ calculation provides the controlled infrared test.

The apparent equality of the three profiles at $a=1$ has a simple
lattice explanation. With the present open-chain convention the site
coordinates are half-integers, so the support condition $|x|<a$
leaves only the two sites $x=\pm1/2$ when $a=1$. On those two sites,
\begin{equation*}
w_2(\pm1/2)=\left(1+\frac18\right)w_1(\pm1/2),
\qquad
w_3(\pm1/2)=\left(\frac34\right)^4,
\end{equation*}
and the multiplicative constants disappear upon normalization. All
three profiles therefore produce exactly the same lattice weight
vector. The equality in the original table is thus a resolution
artifact, not a numerical coincidence. It is also a useful warning:
a width smaller than a few lattice spacings does not resolve the
continuum profile family and should not be used to test profile
universality.

The controlled $M=64$ von Neumann values are
\begin{table}[H]
\centering
\small
\begin{tabular}{c|ccc}
\toprule
$ a $ & $w_1$ & $w_2$ & $w_3$\\
\midrule
1 & $+0.002302$ & $+0.002302$ & $+0.002302$\\
2 & $-0.001214$ & $-0.001722$ & $+0.001520$\\
4 & $-0.003052$ & $-0.003102$ & $-0.002624$\\
\bottomrule
\end{tabular}
\caption{Controlled finite-cutoff spatial corrections at $N=128$,
$\tau=100$, $M=64$, and $\lambda_f=1$. The identical $a=1$ entries
are explained by the unresolved two-site profile described in the
text.}
\label{tab:controlled}
\end{table}

The table is not a continuum extrapolation. It records the
finite-cutoff response of the coupled lattice; profile dependence
becomes visible once the profile occupies several lattice sites. The
physics question is therefore not whether these three numbers should
be identical. The analytic scattering theorem fixes only the
zero-frequency map. Entanglement samples the covariance over a band of
lattice momenta, so the finite-cutoff calculation is precisely where
higher spatial moments of $w$ can remain visible. Agreement of the
profiles is expected only after the low-energy/zero-width scaling limit
is taken at fixed scattering angle.

\subsection{Is the numerical response converged in time?}
\label{sec:timestep}

For the representative controlled run $N=128$, $a=2$, $w_1$, $M=64$,
$\lambda_f=1$, and $\tau=100$, we evolved over $[-4\tau,4\tau]$.
Refining the number of midpoint/Cayley steps gives
\begin{table}[H]
\centering
\small
\begin{tabular}{rcc}
\toprule
number of steps & $\delta t$ & $\Delta S_1(M=64)$\\
\midrule
80 & 10 & $-0.0012139$\\
800 & 1 & $-0.0012196$\\
4000 & 0.2 & $-0.0012186$\\
\bottomrule
\end{tabular}
\caption{Time-step refinement for the controlled $N=128$, $a=2$,
$w_1$, $M=64$, $\lambda_f=1$, $\tau=100$ calculation. The total
spread is about $0.4\%$ of the correction.}
\label{tab:timestep}
\end{table}
This check supports the numerical value at fixed $\tau$; it is not a
switching-time convergence test. The switching-time dependence is
assessed separately with the timestep held fixed.

\subsection{Does slowing the switch change the physical response?}
\label{sec:tau-convergence}

The zero-memory approximation is controlled by
$\epsilon_{\rm mem}=a/\tau$. To verify that this is small in
practice, we perform the representative $N=128$, $a=2$, $w_1$,
$M=64$ calculation at three switching times $\tau=50,100,200$,
holding all other parameters fixed. The results are
\begin{table}[H]
\centering
\small
\begin{tabular}{c|ccc}
\toprule
$\tau$ & 50 & 100 & 200\\
\midrule
$\Delta S_1(M=64)$ & $-0.001178$ & $-0.001218$ & $-0.001219$\\
\bottomrule
\end{tabular}
\caption{Switching-time scan at fixed $N=128$, $a=2$, $w_1$,
$M=64$, $\lambda_f=1$. The total spread is about $4\times10^{-5}$; the change from $\tau=100$
to $200$ is only $10^{-6}$.}
\label{tab:tau-convergence}
\end{table}
The refined values show convergence of the finite-cutoff calculation
with switching time: the change from $\tau=100$ to $200$ is only
$10^{-6}$. This scan therefore serves as an adiabatic-convergence check
for the final-Hamiltonian response. It is evidence that the reported
spatial benchmark is probing the intended post-switch ground-state
response, not evidence that the continuum finite-memory correction is
exactly zero. It does not constitute a rigorous
error bound on the continuum Bogoliubov kernel; such a bound would
require solving the finite-memory time-dependent problem directly.

\subsection{Does the weak response organize in the physical scattering angle?}
\label{sec:renyiscaling}

The parity theorem makes the weak-shutter coefficient $Q_\alpha$ the
natural finite-cutoff observable. Rather than fitting a four-point
entropy difference to $A_\alpha\log M+B_\alpha$, we determine the
response directly from small angle changes,
\begin{equation}
Q_\alpha(M,a)=\frac{\Delta S_A^{(\alpha)}}{\Delta\theta^2}+O(\Delta\theta^2).
\label{eq:Qnumerical}
\end{equation}
The continuum question is then whether, after a controlled
lattice-spacing extrapolation and with the interval represented by
its chord length
\begin{equation*}
\ell_M=\frac{2(N+1)}{\pi}\sin\frac{\pi M}{N+1},
\end{equation*}
the quantity $Q_\alpha(\ell_M,a)$ approaches a universal
large-$\ell_M$ form. The present data are not sufficient to determine
that asymptotic form. In particular, a fit of a small difference of
two large entropies over four interval sizes can produce a spurious
slope and should not be interpreted as a continuum logarithmic
coefficient.

The independent weak-shutter scan at $N=128$, $a=2$, $w_1$, $M=64$,
$\tau=100$, and 4000 midpoint/Cayley steps gives, for
$\lambda_f=0.1,0.2,0.4,0.8,1.0$, the numerical ratios
\begin{equation*}
\frac{\Delta S_1}{\Delta\theta^2}\simeq-1.584\times10^{-3},\,-1.584\times10^{-3},\,-1.585\times10^{-3},\,-1.602\times10^{-3},\,-1.625\times10^{-3}.
\end{equation*}
The near constancy through $\lambda_f=0.4$ provides a direct
numerical check of the quadratic weak-shutter regime, while the
roughly $3\%$ departure at $\lambda_f=1$ indicates the expected
higher-order correction. We use this scan as the numerical validation
of Eq.~\eqref{eq:weakquadratic}, rather than interpreting the four-point logarithmic fits as continuum
asymptotics.

For the same data, using the bare coupling instead of the physical
angle gives a visibly more parameter-dependent ratio. For example, at
$\lambda_f=0.1,0.4,0.8$,
\begin{table}[H]
\centering
\small
\begin{tabular}{c|cc}
\toprule
$\lambda_f$ & $\Delta S_1/\Delta\theta^2$ ($\times10^{-3}$) & $\Delta S_1/\lambda_f^2$ ($\times10^{-3}$)\\
\midrule
0.1 & $-1.583$ & $-1.578$\\
0.4 & $-1.584$ & $-1.504$\\
0.8 & $-1.602$ & $-1.320$\\
\bottomrule
\end{tabular}
\caption{Comparison of the physical matching variable
$\Delta\theta$ with the bare coupling $\lambda_f$. Over the full
weak-shutter scan $\lambda_f=0.1$--$1.0$, the angle-normalized ratio
varies by about $2.6\%$, whereas the coupling-normalized ratio varies
by about $19.5\%$.}
\label{tab:theta-vs-lambda}
\end{table}

As a separate geometric check, we use the chord length $\ell_M$
rather than $M$ whenever interval size is compared with continuum
CFT scaling. For the vacuum on the doubled lattice, the fit over
$M=4,\ldots,80$ gives a slope $0.3312$ with maximum residual
$0.0056$ against $\log\ell_M$, compared with slope $0.2638$ and
maximum residual $0.101$ against $\log M$. This confirms that the
chord variable is the appropriate finite-size coordinate for the
vacuum benchmark. This is consistent with the use of physical interval
length in the CFT form of Eq.~\eqref{eq:CFTrenyi}. The shutter-induced
correction itself remains too small on the available lattice to support
an asymptotic logarithmic fit.

At small angle, a direct $Q_\alpha(M)$ diagnostic is more
informative than a logarithmic fit. For $\lambda_f=0.1$, $a=2$, and
$w_1$, the static final-ground-state calculation gives
\begin{table}[H]
\centering
\small
\begin{tabular}{c|rr}
\toprule
$M$ & $Q_1$ & $Q_2$\\
\midrule
8 & $-0.0271$ & $-0.0005$\\
16 & $-0.0130$ & $+0.0059$\\
32 & $-0.0057$ & $+0.0058$\\
64 & $-0.0016$ & $+0.0032$\\
112 & $\simeq0$ & $+0.0005$\\
\bottomrule
\end{tabular}
\caption{Exploratory weak-shutter coefficients $Q_1$ and $Q_2$ on
the $N=128$ lattice at $a=2$, $w_1$, using the small-$\lambda_f$
final-ground-state response. The values decrease toward zero as the
interval approaches the full finite system; no continuum large-$L$
inference is made.}
\label{tab:Q-M}
\end{table}
For $M=8,16,32,64$, the product $M Q_1$ is approximately
$-0.22,-0.21,-0.18,-0.10$, respectively; the apparent decrease is
consistent with strong finite-size and boundary effects rather than
evidence for an asymptotic law.

A separate finite-size check at fixed lattice width $a=2$ is more
informative. Using static final-ground-state responses at
$\lambda_f=0.1$ and $w_1$, the product $M Q_1$ approaches a common
finite-cutoff plateau as $N$ increases:
\begin{table}[H]
\centering
\small
\begin{tabular}{c|rrrr}
\toprule
$N$ & $M=8$ & $M=16$ & $M=32$ & $M=64$\\
\midrule
128 & $-0.216$ & $-0.208$ & $-0.183$ & $-0.102$\\
256 & $-0.215$ & $-0.211$ & $-0.205$ & $-0.181$\\
512 & $-0.214$ & $-0.210$ & $-0.209$ & $-0.204$\\
\bottomrule
\end{tabular}
\caption{Finite-size collapse of the weak-shutter von Neumann
coefficient at $\lambda_f=0.1$, $a=2$, and $w_1$, using static
final-ground-state responses. For $M\ll N$, $M Q_1$ approaches
approximately $-0.21$ as $N$ increases. This is a finite-cutoff
ground-state diagnostic: the approach toward $-0.21$ is consistent
with a $1/M$ decay at fixed lattice ultraviolet cutoff for
$M\ll N$, but it is not a continuum infrared exponent and not a
test of dynamical memory.}
\label{tab:Ncollapse}
\end{table}
The collapse is consistent with $Q_1\propto M^{-1}$ over the
displayed finite-size window, but we deliberately do not promote
that observation to a continuum exponent: the limit
$a_{\rm UV}\to0$ and the corresponding physical rescaling of $M$
have not been taken.

\subsection{Why does finite width remain visible?}
\label{sec:profilediagnostic}

The three normalized profiles differ not only in their pointwise
shape but also in their moments. Their second and fourth moments are
\begin{table}[H]
\centering
\small
\begin{tabular}{c|cc}
\toprule
 & $\int du\,u^2w(u)$ & $\int du\,u^4w(u)$\\
\midrule
$w_1$ & 0.1581 & 0.0530\\
$w_2$ & 0.1711 & 0.0598\\
$w_3$ & 0.0909 & 0.0210\\
\bottomrule
\end{tabular}
\caption{Second and fourth moments of the three normalized shutter
profiles.}
\label{tab:profile-moments}
\end{table}
The profile form factor gives a compact quantitative measure of the
finite-width structure sampled at a dimensionless momentum $ka$,
\begin{equation*}
\widehat w(k)=\int_{-1}^{1}du\,w(u)e^{iku}.
\end{equation*}
For the even profiles used here, $\widehat w$ is real. At the representative
dimensionless momenta $ka=\pi/2$ and $\pi$, we obtain
\begin{table}[H]
\centering
\small
\begin{tabular}{c|cc}
\toprule
 & $\widehat w(\pi/2)$ & $\widehat w(\pi)$\\
\midrule
$w_1$ & 0.8179 & 0.4065\\
$w_2$ & 0.8036 & 0.3657\\
$w_3$ & 0.8930 & 0.6278\\
\bottomrule
\end{tabular}
\caption{Dimensionless profile form factors for the three normalized
continuum profiles. The entries are $\widehat w(q)=\int_{-1}^1du\,w(u)e^{iqu}$
at $q=\pi/2$ and $\pi$. They illustrate quantitatively why profiles
with the same integrated coupling can retain different finite-$a$ lattice
responses away from the zero-frequency limit.}
\label{tab:profile-form-factor}
\end{table}
These values are not used as a one-parameter fit: the covariance samples a
range of lattice momenta, so the full form factor, rather than one selected
value, controls finite-resolution profile dependence.

The finite-lattice entropy differences do not follow the second
moment monotonically: $w_3$ has the smallest second moment but gives
the largest correction in several of the controlled comparisons.
Thus a single effective-width parameter is insufficient to explain
the observed profile dependence. This is physically reasonable. The
static scattering theorem constrains the integrated coupling at
$\omega=0$, whereas the entropy samples a range of lattice momenta;
higher moments of the profile, and therefore its Fourier transform
away from zero momentum, can enter the finite-$a$ covariance. We
therefore interpret the observed profile dependence as a
finite-resolution effect whose continuum fate remains to be tested
at fixed physical angle and fixed low-energy scale.

\subsection{Can a frequency partition see the infrared correlations?}
\label{sec:soft-hard-numerical}

The final entanglement diagnostic asks the most direct infrared question:
if the shutter produces soft correlations, can a frequency-space
bipartition isolate them? The soft/hard calculation was performed in
the free energy-mode basis using complete positive-energy shells. As in the spatial
analysis, these are modes of the doubled lattice regulator, not a
one-to-one realization of the single counter-propagating continuum
Majorana pair. Because the shell index $k$ does not correspond to a
fixed physical $\Omega_c$ when $N$ changes, the sequences below are
retained only as exploratory finite-lattice diagnostics. A continuum
comparison requires holding $\Omega_c$ fixed while reducing the
lattice spacing.

At $N=128$, profile $w_1$, and width $a=4$, the von Neumann entropy
for $k=2,4,8,16,32$ total shells (equivalently $k/2$ per chain) is
approximately
\begin{equation*}
0.2579,\quad0.2488,\quad0.2208,\quad0.1895,\quad0.1305.
\end{equation*}
There is no general monotonicity theorem for this bipartition. When
$k$ is increased, the subsystem projector changes and so does the
physical cutoff $\Omega_c$; at the same time the complementary
sector loses modes. The entropy is therefore not an entropy of a
nested spatial region, for which simple monotonicity intuition might
apply. The observed decrease with $k$ is consequently a property of
this finite-lattice sequence, not by itself a universal infrared
law.

A fixed-physical-$\Omega_c$ comparison between $N=128$ and $N=256$ does not
show numerical convergence at the present resolution. Using the same protocol,
the von Neumann entropies at $\Omega_c=0.0855,0.1215,0.1703$ change from
$0.240,0.216,0.2208$ at $N=128$ to $0.277,0.265,0.2352$ at $N=256$.
The decrease with shell count is therefore retained as a finite-lattice diagnostic,
not identified as a lattice artifact or as a converged continuum effect. A genuine
continuum comparison requires fixed physical $\Omega_c$ together with controlled
spatial and infrared scaling.

The Rényi response at $k=8$, $a=4$, for the three profiles is shown
in Table~\ref{tab:soft}. The profile dependence is modest but
nonzero at this finite lattice resolution.

\begin{table}[H]
\centering
\small
\begin{tabular}{ccrrrrr}
\toprule
profile & $a$ & $S_{1/2}$ & $S_1$ & $S_2$ & $S_3$ & $S_\infty$\\
\midrule
$w_1$&4&0.7454&0.2208&0.0796&0.0600&0.0400\\
$w_2$&4&0.7426&0.2189&0.0787&0.0594&0.0396\\
$w_3$&4&0.7661&0.2341&0.0859&0.0648&0.0432\\
$w_1$&8&0.6593&0.1721&0.0583&0.0439&0.0293\\
$w_2$&8&0.6534&0.1688&0.0569&0.0429&0.0286\\
$w_3$&8&0.6986&0.1941&0.0678&0.0511&0.0340\\
\bottomrule
\end{tabular}
\caption{Soft/hard frequency-mode entropies at $N=128$ for total
shell count $k=8$, i.e.\ four complete positive-energy shells per
chain, corresponding to $\Omega_c=\varepsilon_4\simeq0.1703$. The
values are finite-lattice diagnostics; the vacuum values are zero
within numerical precision.}
\label{tab:soft}
\end{table}

At $N=256$, $w_1$, $a=8$, the corresponding von Neumann entropies
for $k=4,8,16,32$ total shells are approximately
$0.3494,0.2496,0.1857,0.1224$. The shell energies are given
analytically by Eq.~\eqref{eq:shell-energy}; because the discrete
spectra at $N=128$ and $N=256$ do not contain exactly the same
cutoff values, a fixed-$\Omega_c$ comparison requires interpolation
or nearest-shell matching. We therefore retain these values as
finite-size diagnostics rather than as a continuum scaling test.

\subsection{What the numerical results establish}

The pilot calculation establishes four limited but useful facts.
\begin{enumerate}[label=(\roman*)]
\item The full Majorana covariance gives a stable numerical route
to all three partitions.
\item The spatial, two-chain, and soft/hard entropies are genuinely
different observables even though they are computed from the same
evolved Gaussian state.
\item The finite-lattice Rényi entropies are partition dependent;
the data are not used to extract a continuum $c_{\rm eff}$ or to
test the ordinary-CFT factor $1+1/\alpha$.
\item Finite-width and profile effects remain visible at the lattice
scales used here; consequently, profile independence must be
demonstrated by extrapolation rather than assumed from the
zero-frequency scattering theorem.
\end{enumerate}

The calculation does \emph{not} establish a continuum large-$L$
law, an exact entropy coefficient, or an infrared divergence of the
spatial entropy. Nor does the $N=128,\tau=100$ evolution
numerically resolve the continuum window
$\omega_{\min}\ll\tau^{-1}$ required to test the global infrared
logarithm. It does not replace the exact finite-memory time-dependent
scattering problem. Its controlled numerical content is instead the
finite-cutoff covariance construction, the general even-parity
constraint, the weak-shutter quadratic response, the switching-time
scan, and finite-width/profile diagnostics.

\section{Finite-width universality at the entanglement level}
\label{sec:universality}

The preceding sections establish two different facts: finite-width
scattering is universal at low frequency when profiles are compared
at fixed $\theta_\lambda$, while the finite-lattice entropy still
shows profile dependence. We now formulate the additional scaling
limit needed to promote one-particle universality into an
entanglement statement. If two profiles have different bare
couplings but the same asymptotic angle, they should be compared at
fixed $\theta_\lambda$. The continuum universality question is then
whether
\begin{equation}
\Delta S_A^{(\alpha)}(a,L;w)
\longrightarrow
\Delta S_A^{(\alpha)}(0,L)
\label{eq:ent-univ}
\end{equation}
when $a/L\to0$ at fixed angle, and whether the limit is independent
of $w$.

A sufficient route would be a bound such as
\begin{equation}
\norm{\Gamma_A^{(a)}-\Gamma_A^{(0)}}_1
\le C_A a^\eta,
\qquad \eta>0,
\label{eq:trace-sufficient}
\end{equation}
for fixed $L$ and regulator. This would imply convergence of the
finite-dimensional reduced density matrices and therefore of the
entropy. The present scattering estimate, by itself, does not
provide Eq.~\eqref{eq:trace-sufficient}; it controls the one-particle
on-shell map.

The large-interval problem is therefore a two-stage limit:
\begin{equation*}
\lim_{L/a_{\rm UV}\to\infty}
\lim_{a/L\to0}
\Delta S_A^{(\alpha)}(a,L),
\end{equation*}
with the order of infrared and ultraviolet limits specified
explicitly. If a direct frequency regulator $\mu$ is used, the
order of limits must include $\mu\to0$ as well.

\subsection{What is and is not determined}
\label{sec:what-established}

It is useful to separate the analytic infrared statements, the new finite-lattice
benchmarks, and the limits that remain open.

\paragraph{Established analytically.}
\begin{itemize}
\item The exact zero-frequency scattering matrix $S_a(0)=R(\theta_\lambda)$ is
profile-independent (Proposition 3.1).
\item The finite-frequency correction is bounded by $Ca|\omega|$ in the Frobenius
norm (Eq.~\eqref{eq:Sbound}).
\item The point-contact Bogoliubov kernel has the universal inverse-sum-frequency
singularity, and its squared Hilbert--Schmidt norm grows logarithmically
(Eq.~\eqref{eq:HSlog}).
\item All fixed-cutoff entropies in the spatial, chain, and soft/hard partitions
are even in $\Delta\theta$ (Eq.~\eqref{eq:generalparity}).
\end{itemize}

\paragraph{Established numerically at finite cutoff.}
\begin{itemize}
\item The doubled-lattice ground-state overlap exhibits Anderson
orthogonality (AO): the overlap follows the predicted
$\gamma_{\rm AO}=(\Delta\theta/2\pi)^2$ scaling within the quoted
finite-width and profile corrections.
\item The excitation number $N_{\rm exc}=\|\beta\|_{\rm HS}^2$ shows
logarithmic finite-size growth whose coefficient follows the parameter-free
$2C_{\rm qp}=4\sin^2(\Delta\theta/2)/\pi^2$ prediction in the tested
sudden and smooth-switch windows, including the resolved $\tau=100$
calculation (Sec.~\ref{sec:IRwindow}).
\item The half-chain defect geometry gives a parameter-free test of the
Eisler--Peschel effective central charge using
$s=\sech\lambda_f=\cos\theta_{\lambda_f}$, with agreement at the
percent level across the tested transmission range and improved agreement
for $a=4$ (Sec.~\ref{sec:defect-interpretation}).
\item The switching-time scan supports adiabatic convergence of the
controlled finite-cutoff final-ground-state response; it is not a rigorous
bound on the finite-memory continuum correction.
\end{itemize}

\paragraph{Open.}
\begin{itemize}
\item The continuum large-$L$ asymptotic of $\Delta S_A^{(\alpha)}(L)$
for an interval whose cuts do not coincide with the shutter is not determined.
\item The finite-memory correction to the time-dependent Bogoliubov kernel
is not bounded rigorously.
\item The soft/hard frequency-partition data remain finite-lattice diagnostics;
fixed-physical-cutoff convergence has not been established.
\end{itemize}

\subsection{Future numerical directions}
\label{sec:future-numerical}
The main remaining numerical target is a controlled continuum spatial-entropy
scaling study at fixed physical scattering angle and fixed physical
ultraviolet prescription. A second analytic target is a rigorous estimate of
the finite-memory correction to the time-dependent scattering kernel. Neither
is required for the finite-cutoff infrared benchmarks established here.

\section{Universal infrared scaling: orthogonality, Loschmidt dynamics, and a logarithmic-CFT consistency check}
\label{sec:universal-scaling}

The scattering angle becomes substantially more informative when
it is treated as a boundary-condition-changing parameter rather
than only as an input to the Bogoliubov kernel. A sudden change of
a local scattering condition is the canonical setting for Anderson
orthogonality (AO) and boundary-condition-changing (BCC) operators: the
many-body overlap becomes a finite-size probe of the scaling
dimension associated with the change of boundary data
\cite{Affleck1996,Anderson1967,NozieresDeDominicis1969}. The same
logic underlies conformal descriptions of local quenches and their
Loschmidt echoes \cite{StephanDubail2011}.

For the present two-Majorana problem the real rotation can be
combined into a complex fermionic channel,
\begin{equation*}
\psi=\frac{\chi_R+\ii\chi_L}{\sqrt2},
\end{equation*}
under which the zero-frequency rotation acts as a phase shift. The
change between the early and late scattering conditions is therefore
characterized by the single angle
\begin{equation*}
\Delta\theta=\theta_+-\theta_-.
\end{equation*}
For the present normalization it is useful to make the doubling
factor explicit. In a complex-fermion representation, the rotation
eigenchannels acquire opposite phase shifts $\pm\delta$ with
\begin{equation*}
\delta=\frac{\Delta\theta}{2}.
\end{equation*}
For one continuum Majorana pair, the Majorana constraint identifies
the conjugate channels, giving one half of the corresponding
complex-fermion orthogonality exponent,
\begin{equation*}
\gamma_{1\,\mathrm{pair}}=\frac{\delta^2}{2\pi^2}
=\frac{(\Delta\theta)^2}{8\pi^2}.
\end{equation*}
The open-chain regulator used in this paper contains two low-energy
sectors, near $k=0$ and $k=\pi$. In the continuum-sector identification
each sector is associated with the same counter-propagating Majorana pair,
so the doubled-lattice exponent is twice the single-pair value. For the
open chain this sector identification is not a first-principles derivation
of the boundary exponent because the finite system has standing-wave modes;
the doubling factor is independently validated by the numerical agreement
with the parameter-free prediction below,
\begin{equation}
\boxed{\gamma_{\rm AO}^{\rm(doubled)}
=2\gamma_{1\,\mathrm{pair}}
=\left(\frac{\Delta\theta}{2\pi}\right)^2.}
\label{eq:gammaAO}
\end{equation}
This is the exponent tested numerically below. It should not be
identified with the exponent of a single undoubled continuum
Majorana pair, which is smaller by a factor of two in this
convention. The derivation also makes clear why the result is
quadratic in the phase-shift mismatch: it is the sum of the
universal BCC/orthogonality contributions of the low-energy
sectors. The overlap and $C_{\rm qp}$ remain independent
observables; their common dependence on $\Delta\theta$ is not an
identity. More specifically, $C_{\rm qp}$ is the coefficient of a
logarithmic infrared divergence generated by a frequency integral
of the anomalous kernel, whereas $\gamma_{\rm AO}$ is the exponent
of a finite-size power law associated with a boundary-condition-changing
scaling dimension. Their common quadratic small-angle behavior
follows from the same scattering phase shift, but they are not the
same observable.

\section{Resolving the infrared window: the Bogoliubov logarithm on the lattice}
\label{sec:IRwindow}

The infrared coefficient $C_{\rm qp}$ multiplies
$\log(\omega_*/\mu)$, so a lattice measurement requires a resolved
window $\varepsilon_1\ll\omega\ll\omega_*$. For the open chain,
$\varepsilon_1\simeq\pi/(N+1)$. A sudden switch has an upper scale
set by the finite shutter width, $\omega_*\sim1/a$, whereas a smooth
switch of duration $\tau$ requires $\varepsilon_1\tau\ll1$ before
the adiabatic soft-mode region is separated from the logarithmic window.

\paragraph{Observable.}
For a pure Gaussian state $|\Psi_-\rangle$ and the ground state
$|\Psi_+\rangle$ of the post-shutter Hamiltonian, the number of
$H_+$ quasiparticles in $|\Psi_-\rangle$ is
$N_{\rm exc}=\|\beta\|_{\rm HS}^2$. With
$\Gamma_\pm^2=-I$ on $2N$ Majoranas,
\begin{equation}
N_{\rm exc}=\tfrac12 N+\tfrac14\Tr(\Gamma_+\Gamma_-)
=\frac18\|\Gamma_+-\Gamma_-\|_{\Frob}^2 .
\label{eq:Nexc}
\end{equation}
The infrared logarithm therefore predicts
\begin{equation}
\frac{dN_{\rm exc}}{d\ln N}=n_c C_{\rm qp}
=n_c\frac{2}{\pi^2}\sin^2\frac{\Delta\theta}{2},
\label{eq:Nexc-slope}
\end{equation}
where $n_c=2$ counts the two low-energy lattice sectors just defined,
one descending from the bulk region $k\simeq0$ and the other from
$k\simeq\pi$. The coefficient $C_{\rm qp}$ itself refers to one
continuum counter-propagating Majorana pair. At finite angle the dependence is not
exactly quadratic in $\Delta\theta$.

\paragraph{Sudden switch.}
We compute $\Gamma_\pm$ from the exact ground states of $H_0$ and
$H_0+\lambda_fH_{\rm sh}$ with $a=2$ and $w_1$. For
$N=64,\ldots,2048$ and, for $\lambda_f=0.4,1,3$, also $N=3072$, the
logarithmic slopes are
\begin{table}[H]
\centering
\small
\begin{tabular}{cccccc}
\toprule
$\lambda_f$&$\Delta\theta$&$N_{\rm exc}(3072)$&slope ($2048\to3072$)&
$\frac{4}{\pi^2}\sin^2\frac{\Delta\theta}{2}$&$\Delta\theta^2/\pi^2$\\
\midrule
0.4&0.3897&0.13944&0.01509&0.01520&0.01539\\
1.0&0.8658&0.66684&0.06889&0.07132&0.07595\\
3.0&1.4713&1.93027&0.17191&0.18251&0.21933\\
\bottomrule
\end{tabular}
\caption{Logarithmic slope of the excitation number after a sudden
switch for $a=2$, $w_1$. The measured slopes follow the finite-angle
$2C_{\rm qp}=4\sin^2(\Delta\theta/2)/\pi^2$ dependence more closely
than its quadratic small-angle approximation. The largest-angle case
still shows finite-width convergence with increasing $a$.}
\label{tab:Nexc}
\end{table}
At weak angle, $\lambda_f=0.2$, the fitted slope over
$N=512$--$2048$ is $0.00398$, compared with the prediction $0.00399$.
For $\lambda_f=3$ the successive local slopes
$0.1616,0.1661,0.1689,0.1705,0.1715$ show residual finite-size
convergence. Increasing the width improves the agreement: for
$\lambda_f=3$ the measured/predicted ratios over $N=512\to1024$ are
$0.920$, $0.934$, and $0.994$ for $a=1,2,4$, respectively; for
$\lambda_f=1$ they are $0.969$, $0.965$, and $0.999$.

\paragraph{Smooth switch.}
For the smooth protocol at $\tau=10$, with $\lambda_f=1$ and $a=2$,
$N/(\pi\tau)$ ranges from $4$ to $33$ for $N=128$--$1024$. The
excitation numbers are $0.1339,0.1839,0.2324,0.2804$, with
successive-doubling slopes $0.0722,0.0699,0.0693$, consistent with
$0.0713$.

A more stringent resolved-window calculation is obtained at
$\tau=100$, $\lambda_f=1$, $a=2$, $w_1$. We evolve the covariance over
$[-4\tau,4\tau]$ using a Cayley step $\delta t=10$ and a banded
solver for the pentadiagonal generator in the interleaved $R_j,L_j$
ordering. At $N=256$, refinement $\delta t=10,5,2.5,1$ gives
$N_{\rm exc}=0.005283,0.005268,0.005267,0.005266$, demonstrating
$0.3\%$ step-size convergence. The implementation reproduces the
$\tau=10$, $N=256$ value $0.183916$.

The window parameter is
$\varepsilon_1\tau=\pi\tau/(N+1)$. The resulting data are
\begin{table}[H]
\centering
\small
\begin{tabular}{rccc}
\toprule
$N$&$\varepsilon_1\tau$&$N_{\rm exc}$&local slope\\
\midrule
384&0.816&0.02667&\\
512&0.612&0.05249&0.0898\\
768&0.409&0.09155&0.0963\\
1024&0.307&0.11607&0.0853\\
1536&0.204&0.14698&0.0762\\
2048&0.153&0.16780&0.0724\\
3072&0.102&0.19652&0.0708\\
\bottomrule
\end{tabular}
\caption{Excitation number after a smooth $\tau=100$ switch. The local
slope approaches the predicted doubled-lattice coefficient
$2C_{\rm qp}=0.0713$ as $\varepsilon_1\tau$ decreases. The final local
slope is $0.0708$ at $N=2048\to3072$, while a fit over
$N=1536$--$3072$ gives $0.0714$.}
\label{tab:smooth100}
\end{table}
The approach is consistent with the crossover from the adiabatic regime,
where modes below the switching scale are strongly suppressed, to the
infrared logarithmic regime. Within the tested range, the same coefficient
is obtained for the sudden and smooth protocols.

\paragraph{Relation to the orthogonality exponent.}
For independent Bogoliubov pairs with small squeezing parameters $r_j$,
$N_{\rm exc}=2\sum_j\sinh^2r_j$ whereas
$-\ln|\langle\Psi_-|\Psi_+\rangle|
=\tfrac12\sum_jr_j^2+O(r^4)$. Hence
\begin{equation*}
N_{\rm exc}=4[-\ln|\langle\Psi_-|\Psi_+\rangle|]+O(r^4)
\end{equation*}
with the factor four arising from pair occupation and the overlap
normalization. At small angle the logarithmic slopes therefore satisfy
$2C_{\rm qp}=4\gamma_{\rm AO}=(\Delta\theta/\pi)^2$ for the
doubled lattice. At finite angle, the excitation coefficient is instead
proportional to $\sin^2(\Delta\theta/2)$ whereas the orthogonality
exponent is proportional to $(\Delta\theta/2)^2$.

\paragraph{Scope.}
The sudden-switch calculation reaches $N=3072$. The smooth-switch tests
reach $N=1024$ for $\tau=10$ and $N=3072$ for $\tau=100$. The latter has
$\varepsilon_1\tau\simeq0.10$ at its largest size. Because the evolved
state is numerically indistinguishable from the final ground state at the
resolution relevant here, this is a ground-state finite-size scaling
measurement of the coupled lattice, not a dynamical measurement of the
continuum Bogoliubov kernel and not an asymptotic continuum limit. The exact finite-memory correction to the continuum scattering
kernel remains uncomputed.

\subsection{Ground-state overlap as a BCC observable}
\label{sec:overlap}

Here and below $N$ denotes the number of lattice sites on each
Majorana chain; the doubled covariance therefore has dimension
$2N$. Let $|\Psi_-\rangle$ and $|\Psi_+\rangle$ be the finite-size
Gaussian ground states of the pre- and post-shutter lattice
Hamiltonians. Their overlap is computed directly from the two pure
Majorana covariance matrices. For the covariance convention of
Sec.~\ref{sec:gaussian},
\begin{equation}
\log |\langle\Psi_-|\Psi_+\rangle|
=\frac14\log\det\left[\frac{I-\Gamma_-\Gamma_+}{2}\right].
\label{eq:gaussian-overlap}
\end{equation}
The thermodynamic prediction is
\begin{equation}
|\langle\Psi_-|\Psi_+\rangle|
\sim N^{-\gamma_{\rm AO}}
\left[1+O(N^{-\omega})\right],
\label{eq:AO-scaling}
\end{equation}
with nonuniversal corrections controlled by the lattice cutoff and
the finite width.

The numerical test is deliberately independent of the entropy
calculation. At $a=2$, profile $w_1$, and $\lambda_f=0.4$, the
physical angle is
\begin{equation*}
\theta_f=\arctan(\sinh0.4)=0.38974112,\qquad
\gamma_{\rm AO}=0.00384762.
\end{equation*}
A fit of the actual lattice data at
$N=64,96,128,192,256,384,512$ gives
$x_{\rm AO}^{\rm fit}=0.00383467$ for the representative choice
$a=2,w_1$. Across the width/profile variations $a=1,2,4,8$ and
$w_1,w_2,w_3$ tested at $\lambda_f=0.4$, the extracted exponent
spans $0.003787$--$0.004138$, corresponding to a spread of roughly
$\pm7\%$ around the analytic value $\gamma_{\rm AO}=0.00384762$.
The appropriate conclusion is therefore agreement at the
few-percent level with finite-width/profile/fit-window corrections,
not $0.34\%$ universal accuracy. In the weak-shutter regime the
angular scaling is much sharper: at $\lambda_f=0.2$,
$x_{\rm AO}/\Delta\theta^2=0.0254$, within about $0.3\%$ of
$1/(4\pi^2)$.

\begin{figure}[H]
\centering
\includegraphics[width=0.78\linewidth]{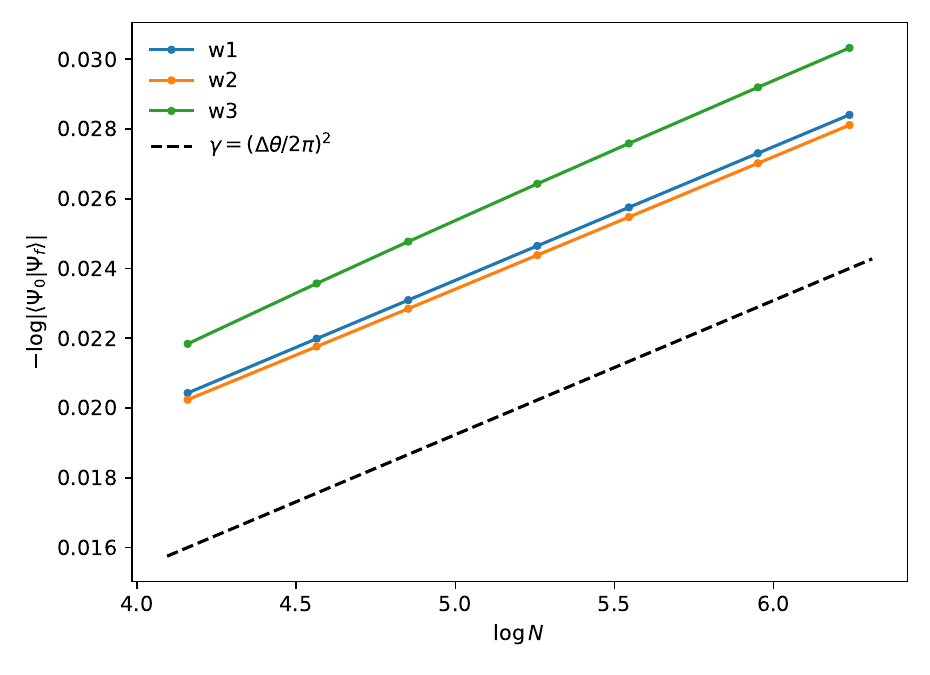}
\caption{Ground-state overlap as a finite-size BCC/orthogonality
diagnostic at $\lambda_f=0.4$, $a=2$. The lattice data for three
smooth profiles follow a common power law. The dashed line uses
the parameter-free infrared prediction
$\gamma_{\rm AO}=(\Delta\theta/2\pi)^2$. Profile dependence is
visible only at the finite-size level and is substantially smaller
than the leading power-law scaling signal.}
\label{fig:overlap-scaling}
\end{figure}

Repeating the fit for $\lambda_f=0.2,0.4,0.6,0.8,1.0$ gives
\begin{table}[H]
\centering
\small
\begin{tabular}{c|cc}
\toprule
$\lambda_f$ & $x_{\rm AO}^{\rm fit}$ & $(\theta_f/2\pi)^2$\\
\midrule
0.2&0.001003&0.0009999\\
0.4&0.003835&0.0038476\\
0.6&0.008037&0.0081416\\
0.8&0.013036&0.0133585\\
1.0&0.018301&0.0189865\\
\bottomrule
\end{tabular}
\caption{Orthogonality exponent scan for the doubled lattice. The
fitted exponent follows the scattering prediction
$\gamma_{\rm AO}=(\Delta\theta/2\pi)^2$ at the few-percent level,
with finite-cutoff corrections increasing at larger angle.}
\label{tab:AOscan}
\end{table}
The agreement is strongest in the weak-to-intermediate regime and
gradually develops the expected finite-cutoff correction as the
angle becomes larger. The relevant point is not the numerical size
of that correction but the collapse when the bare coupling is
replaced by the physical scattering angle.

\begin{figure}[H]
\centering
\includegraphics[width=0.72\linewidth]{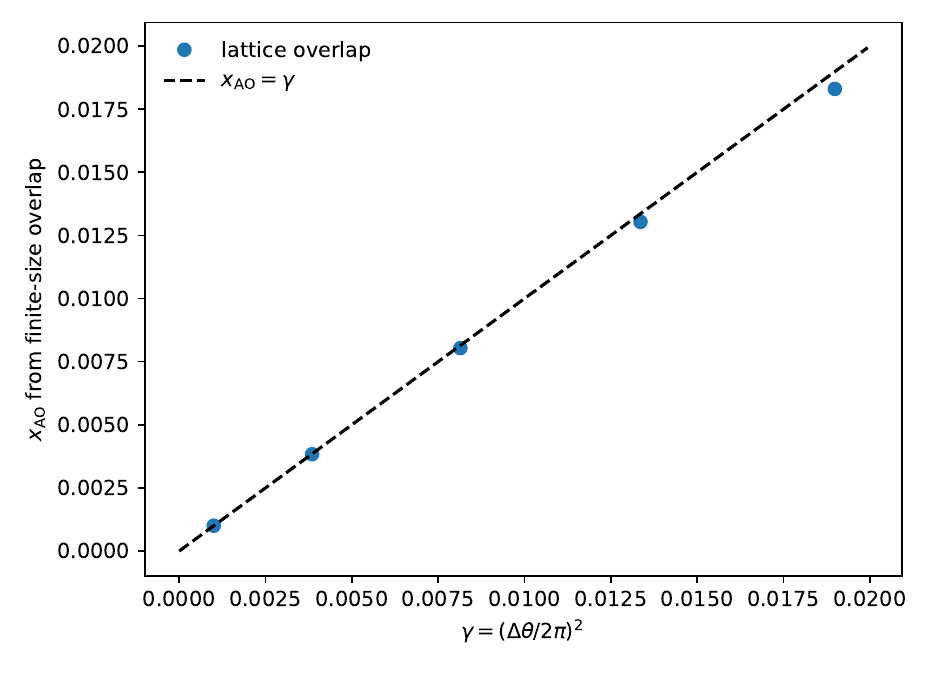}
\caption{Extracted orthogonality exponent versus the scattering
prediction $\gamma_{\rm AO}=(\Delta\theta/2\pi)^2$. The numerical
exponent is obtained from finite-size ground-state overlaps. The
diagonal line is not a fit; it is the analytic scattering
prediction.}
\label{fig:AO-exponent}
\end{figure}

This result gives the paper a direct many-body interpretation of
the infrared angle. The quantity $C_{\rm qp}$ measures the
logarithmic growth of the global anomalous Hilbert--Schmidt norm,
while $x_{\rm AO}$ measures the power-law decay of the many-body
vacuum overlap. They independently exhibit the same leading
small-angle dependence on $\Delta\theta$, but they are different
observables and need not agree at finite angle. In particular,
$C_{\rm qp}=(2/\pi^2)\sin^2(\Delta\theta/2)$ is not equal to
$\gamma_{\rm AO}$ away from the weak-mismatch limit, and the
present numerical runs do not resolve the $\mu\ll\tau^{-1}$ window
needed to measure $C_{\rm qp}$ directly. We therefore regard the
overlap as an independent BCC/orthogonality measurement, not as a
numerical measurement of $C_{\rm qp}$.

\subsection{Loschmidt amplitude and the dynamical BCC test}
\label{sec:loschmidt}

The second independent probe is the sudden-quench Loschmidt
amplitude
\begin{equation}
G(t)=\langle\Psi_-|e^{-\ii H_+t}|\Psi_-\rangle,
\qquad
\mathcal L(t)=|G(t)|^2.
\label{eq:Loschmidt}
\end{equation}
The same Gaussian overlap formula can be applied between
$|\Psi_-\rangle$ and the evolved Gaussian state. For a finite
critical system the power law can only hold before the first
finite-size recurrence, so the fitting window is chosen
conservatively.

For $N=512$, $a=2$, $w_1$, and $\lambda_f=0.4$, an ordinary
least-squares fit of $-\log|G(t)|$ against $\log t$ gives
window-dependent slopes
\begin{equation*}
\begin{array}{c|ccccc}
\toprule
window&$5$--$40$&$5$--$60$&$5$--$80$&$10$--$80$&$20$--$80$\\
\midrule
x_L&0.00773&0.00778&0.00787&0.00793&0.00804\\
\bottomrule
\end{array}
\end{equation*}
compared with $2\gamma_{\rm AO}=0.007695$. Thus the pre-revival
slope is consistent with the BCC prediction at the few-percent
level. The drift is monotonic as the fit window is moved later,
consistent with finite-size effects as the dynamics approaches the
first recurrence at $t\sim N/v$; it is not treated as stochastic
fit noise. We quote the range obtained from the specified fit windows rather
than a single representative value. The
factor of two relative to the ground-state overlap exponent is the
standard BCC/local-quench relation:
$|\langle\Psi_-|\Psi_+\rangle|\sim N^{-\gamma_{\rm AO}}$ whereas
the corresponding Loschmidt amplitude obeys
$|G(t)|\sim t^{-2\gamma_{\rm AO}}$. Time is measured in the lattice
units used throughout the numerical evolution, with low-energy
velocity $v=2$; this velocity changes only the nonuniversal time
scale, not the logarithmic slope. For $N=512$, the observed
recurrence near $t\simeq2(N/2)/v\simeq256$ is consistent with this
finite-size propagation scale.

\begin{figure}[H]
\centering
\includegraphics[width=0.78\linewidth]{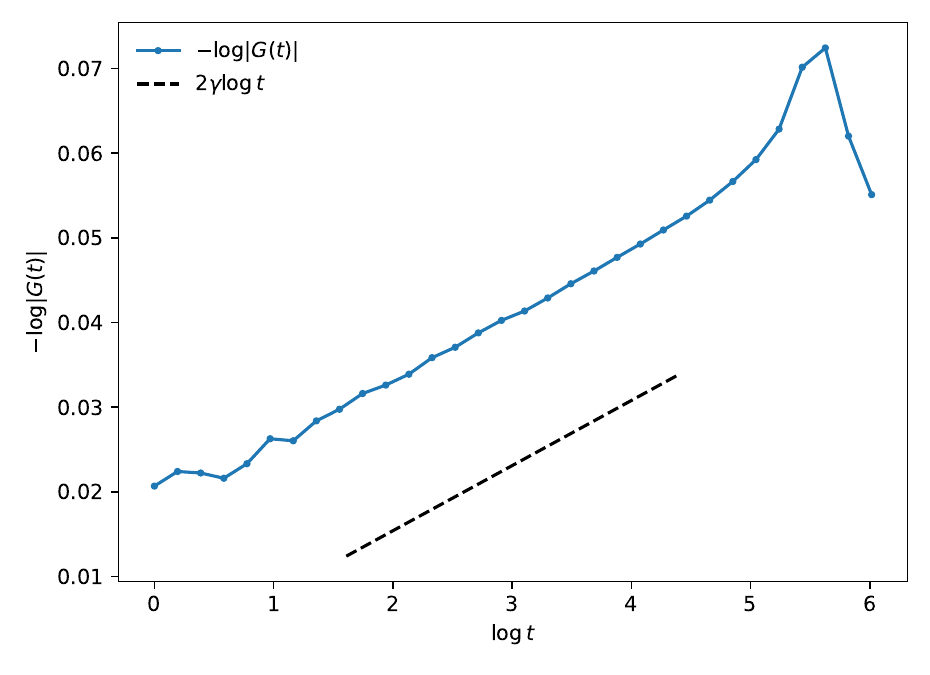}
\caption{Loschmidt amplitude after a sudden switch for $N=512$,
$a=2$, $w_1$, and $\lambda_f=0.4$. The pre-revival data give slopes
in the range $0.00773$--$0.00804$ for the fit windows listed in the
text, compared with $2\gamma_{\rm AO}=0.007695$. The late-time
deviation is a finite-size recurrence effect, not an LCFT
signature.}
\label{fig:loschmidt}
\end{figure}

The Loschmidt calculation is intentionally performed for a sudden
switch even though the controlled shutter dynamics elsewhere in the
paper uses a smooth $\tanh(t/\tau)$ protocol. This separation is
useful: the sudden quench isolates the universal BCC exponent,
while the smooth switch tests the finite-memory problem. A
finite-$\tau$ dynamical exponent should be compared with the
sudden-quench result only after the $a/\tau\to0$ limit is
controlled.

\paragraph{Entanglement-spectrum remark.}
The restricted covariance contains more information than any single
entropy. If $\xi_j$ are the positive eigenvalues of $i\Gamma_A$,
the single-particle entanglement energies are
$\varepsilon_j=2\artanh\xi_j$. In the present work this spectrum is
retained as a regulator-level diagnostic only: its low-lying
structure is consistent with the doubled two-chain construction,
but no quantitative continuum conformal coefficient is extracted
from it. We therefore do not use the entanglement spectrum as an
independent scaling claim.

\subsection{Continuum two-point function and a direct logarithmic-CFT consistency check}
\label{sec:LCFT-correlator}

We emphasize at the outset that a free-fermion lattice regulator
with a Hermitian Hamiltonian is \emph{not} expected to exhibit a
logarithmic conformal field theory (LCFT) Jordan structure. The scaling generator $L_0$ of the free
Majorana fixed point is diagonalizable on the lattice Hilbert space
(up to non-generic degeneracies that are not sampled here), so a
nonzero logarithmic partner cannot appear. The diagnostic below is
therefore a consistency check on the ordinary free-fermion exponent
$p=2$, not a genuine search for an LCFT.

A logarithmic conformal interpretation makes a strong local
statement: after the correct continuum operator has been identified,
a logarithmic partner can produce a two-point function of the form
\begin{equation}
C_{\log}(r)=\frac{A+B\log(r/a_{\rm UV})}{r^{2h}},
\label{eq:logcorr}
\end{equation}
with a nonzero coefficient $B$ that survives the scaling limit and
is associated with an independently identified partner operator. In
a rank-two LCFT the scaling generator acts non-semisimply,
\begin{equation}
L_0|\phi\rangle=h|\phi\rangle,
\qquad
L_0|\psi\rangle=h|\psi\rangle+|\phi\rangle,
\label{eq:Jordan}
\end{equation}
which is the structural origin of logarithmic correlators
\cite{Gurarie1993,Rasmussen2004,Alishahiha2014}. A logarithm in an infrared
integral is therefore not sufficient evidence for an LCFT.

For the free Majorana fixed point the simplest local scalar
bilinear is
\begin{equation}
\mathcal O(x)=i\chi_R(x)\chi_L(x).
\label{eq:mass-operator}
\end{equation}
Each chiral Majorana field has scaling weight $1/2$, so the bulk
scaling dimension of $\mathcal O$ is
\begin{equation}
\Delta_{\mathcal O}=h+\bar h=1.
\label{eq:mass-dimension}
\end{equation}
Consequently an ordinary fixed-point theory predicts
\begin{equation}
\langle\mathcal O(r)\mathcal O(0)\rangle_{\rm conn}
\propto r^{-2},
\label{eq:ordinary-bilinear}
\end{equation}
whereas an LCFT contribution would require a stable logarithmic
modification of this power law. The use of the connected correlator
is essential because the shutter can generate a nonzero local
expectation value $\langle\mathcal O\rangle$; without subtracting
it, a constant disconnected contribution can masquerade as an
anomalous infrared exponent.

The lattice regulator contains two low-energy sectors (the doubling
discussed in Sec.~\ref{sec:lattice}). The connected correlator must
therefore also be allowed to contain a staggered contribution. For
\begin{equation*}
O_j=i\gamma_{R,j}\gamma_{L,j},
\end{equation*}
we compute, for $r>0$,
\begin{align}
C_{\rm conn}(r)
&=\langle O_jO_{j+r}\rangle
 -\langle O_j\rangle\langle O_{j+r}\rangle\notag\\
&=A r^{-p}+(-1)^r B r^{-q}+\cdots .
\label{eq:lattice-correlator-fit}
\end{align}
The expression is evaluated directly from the covariance matrix
using Wick's theorem; no continuum exponent is imposed in the
numerical calculation. The alternating term is a regulator effect
associated with the second low-energy sector and is precisely why
an even/odd subsequence should not be interpreted as a continuum
correlator by itself.

To test the logarithmic hypothesis without assuming it, we compare
Eq.~\eqref{eq:lattice-correlator-fit} with the enlarged form
\begin{equation}
C_{\rm conn}(r)=
\frac{A+B_{\log}\log(r/a_{\rm UV})}{r^p}
+(-1)^r
\frac{C+D_{\log}\log(r/a_{\rm UV})}{r^q}
+\cdots .
\label{eq:lattice-log-fit}
\end{equation}
The ordinary model has four fit parameters, while the logarithmic
model has six. We compare the fitted residuals and, more
importantly, the stability of the logarithmic coefficient under
changes of $N$, shutter strength, anchor site, and fit window. We
do not use BIC: for deterministic lattice data without a specified
stochastic error model, a BIC difference would quantify fit
conditioning rather than statistical evidence.

The calculation was repeated for $N=128,256,512$ and
$\lambda_f=0,0.2,0.4,0.6$, with $a=2$, profile $w_1$, anchor site
$j=N/2$, and the common fit window
\begin{equation*}
5\le r\le r_{\max},\qquad r_{\max}=\min(102,N/4).
\end{equation*}
The shutter data are compared directly with the vacuum control. For
$\lambda_f=0.4$ the fitted non-staggered exponents are
\begin{equation*}
\begin{array}{c|ccc}
\toprule
 &N=128&N=256&N=512\\
\midrule
\lambda_f=0&1.7872&1.8826&1.9400\\
\lambda_f=0.4&1.7878&1.8934&1.9599\\
\bottomrule
\end{array}
\label{tab:LCFT-vacuum-control}
\end{equation*}
The shutter changes the fitted exponent by at most about one
percent over this sequence, while the finite-size drift toward
$p=2$ is already present in the vacuum. The exponent is
consequently not a sensitive shutter diagnostic. The correlator
itself does respond to the shutter: at $j=N/2$ the maximum relative
change of the connected correlator over the fitted range is about
$8$--$9\%$ for $\lambda_f=0.4$. We therefore regard the exponent
primarily as a consistency check on the continuum operator
dimension, not as evidence for a shutter-induced change of scaling
dimension.

The result also depends on the anchor site, as expected for an
open finite chain. At $N=512$ the vacuum fit gives $p=1.923$ for
$j=N/4$ and $p=2.060$ for an anchor chosen as
$j=N/2-r_{\max}/2$. This sensitivity reinforces the need to specify
the anchor and to compare shutter data against the corresponding
vacuum control.

For the logarithmic test we use the actual LCFT form in
Eq.~\eqref{eq:lattice-log-fit}, namely $\log(r/a_{\rm UV})$, not
$\log\log(r/a_{\rm UV})$. The latter is nearly degenerate with a
small shift of the fitted power over the accessible range and is
therefore not a discriminating LCFT ansatz. At $N=512$ the fitted
ratios $B_{\log}/A$ are approximately $-8\times10^{-8}$ in the
vacuum and $7\times10^{-6}$ for $\lambda_f=0.4$; across the tested
sizes and couplings they remain at or below the $10^{-4}$ level.
The fitted $p$ is essentially unchanged when the logarithmic term
is included. These values are not interpreted as statistically
significant detections because the lattice data are deterministic;
their physical significance is their lack of stability under
regulator and control variations.

The appropriate conclusion is a consistency check rather than a
search: the free-fermion regulator reproduces the ordinary
$r^{-2}$ continuum dimension and, as expected on general grounds,
shows no logarithmic admixture. A genuine LCFT search would require
either a non-unitary lattice model or an interacting deformation,
neither of which is considered here.

\begin{figure}[H]
\centering
\includegraphics[width=0.78\linewidth]{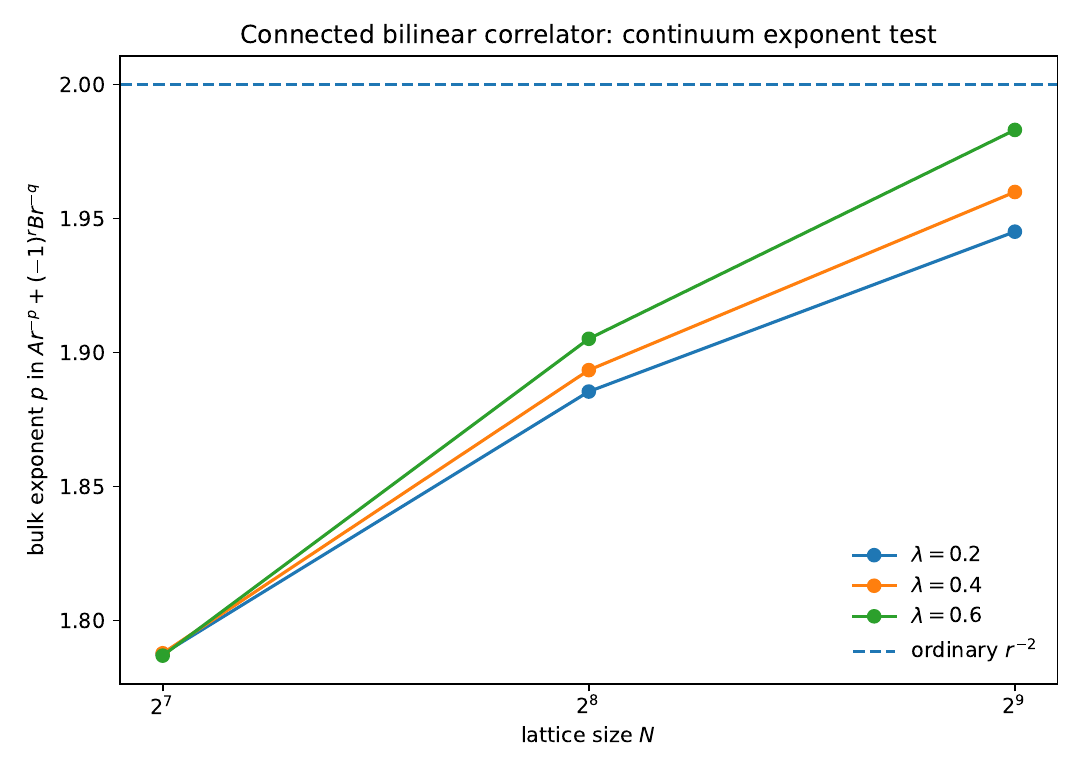}
\caption{Vacuum control for the connected bulk bilinear exponent.
The fits use the anchored two-component form
$A r^{-p}+(-1)^r B r^{-q}$ with $j=N/2$ and
$5\le r\le\min(102,N/4)$. The finite-size drift toward $p=2$ is
already present at $\lambda_f=0$; the shutter changes the fitted
exponent only weakly.}
\label{fig:LCFT-exponent}
\end{figure}

\begin{figure}[H]
\centering
\includegraphics[width=0.78\linewidth]{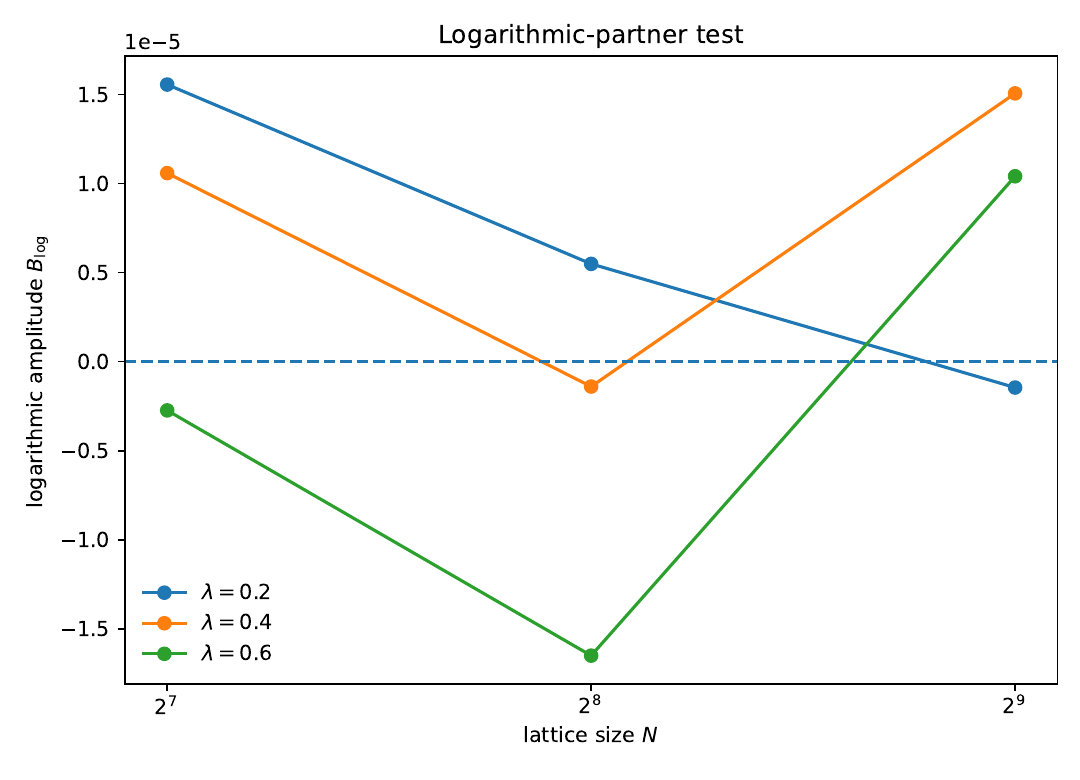}
\caption{Direct logarithmic-partner diagnostic using the correct
$\log(r/a_{\rm UV})$ ansatz. At $N=512$, $B_{\log}/A$ is
$-8\times10^{-8}$ in the vacuum and $7\times10^{-6}$ at
$\lambda_f=0.4$. The small, regulator-dependent amplitudes do not
establish a logarithmic two-point function.}
\label{fig:LCFT-log}
\end{figure}

\subsection{What the infrared logarithms mean}
\label{sec:LCFT}

The preceding calculation separates three mathematically different
statements:
\begin{equation*}
\boxed{\begin{gathered}
\text{IR logarithm in a frequency integral}
\\[-1mm]
\neq\;\text{logarithmic two-point function}
\\[-1mm]
\neq\;\text{LCFT Jordan structure}.
\end{gathered}}
\end{equation*}
The first is established here by the Hilbert--Schmidt memory
integral and by the orthogonality exponent. The second was tested
through Eq.~\eqref{eq:lattice-log-fit}; the present data show no
stable logarithmic amplitude. The third is a statement about the
scaling-limit operator algebra and is not inferred from either
numerical fit. Thus the paper interprets its infrared logarithms as
ordinary memory and orthogonality effects of a unitary free-fermion
problem, while the available two-point data do not provide evidence
for an LCFT logarithmic partner.

This negative result is useful for the spatial-entanglement
problem. It removes one possible mechanism by which a logarithm in
the global infrared observables could have been reinterpreted as an
LCFT-induced entropy enhancement. The remaining continuum spatial
entropy problem is therefore still the ordinary Gaussian
restriction problem described in Sec.~\ref{sec:analytic-spatial},
with its UV/IR order of limits and doubled-lattice caveats intact.

\subsection{Heuristic relation between the soft kernel and spatial entropy}
\label{sec:spatial-heuristic}

The global infrared logarithm and the spatial entropy are related through the
Majorana covariance, but they are not the same functional of the state. In the
weak-shutter regime the logical sequence is
\begin{equation*}
\beta\;\longrightarrow\;\delta\Gamma\;\longrightarrow\;
P_A\delta\Gamma P_A\;\longrightarrow\;\{\delta\xi_j\}\;\longrightarrow\;
\Delta S_A^{(\alpha)}.
\end{equation*}
The first arrow is controlled by the soft Bogoliubov kernel, whereas the last
two arrows involve spatial projection and nonlinear spectral processing. By
contrast, the global quasiparticle number is the unrestricted quadratic trace
\begin{equation*}
N_{\rm exc}=\operatorname{Tr}\,\beta^\dagger\beta
\sim C_{\rm qp}\log(\omega_*/\mu).
\end{equation*}
There is therefore no general sum rule that turns $C_{\rm qp}$ directly into a
spatial entropy coefficient.

The position-space form of the soft kernel makes the distinction more concrete.
At leading order in the mismatch, one has schematically
\begin{equation*}
\widetilde{\beta}(x,y)\sim
-\frac{\Delta\theta}{2\pi}\,\frac{\log(x/y)}{x-y},
\end{equation*}
up to convention-dependent phases and contact/UV terms. If $x=Lu$ and
$y=Lv$ inside an interval of size $L$, then the explicit scale factors cancel:
\begin{equation*}
\widetilde{\beta}(Lu,Lv)\sim
-\frac{\Delta\theta}{2\pi L}\,\frac{\log(u/v)}{u-v}.
\end{equation*}
Consequently the corresponding continuum Hilbert--Schmidt integral over a
fixed scaled region has no automatic factor of $\log L$ from dimensional
rescaling alone. Its remaining sensitivity is concentrated near boundaries,
the diagonal $u=v$, and the ultraviolet regulator. Spatial projection can thus
redistribute or soften the global soft-mode divergence without literally
removing the soft correlations.

A useful perturbative way to state the same point is to write, at fixed cutoff,
\begin{equation*}
\Delta S_A^{(\alpha)}=Q_\alpha(L,a_{\rm UV})\,\Delta\theta^2
+O(\Delta\theta^4),
\end{equation*}
and regard $Q_\alpha$ as a nonlinear functional of the restricted covariance
rather than of $\operatorname{Tr}\beta^\dagger\beta$. The numerical observation
that the centered-interval correction decreases approximately as $1/M$ over
the accessible range is therefore compatible with the absence of an additional
leading logarithm, but it does not establish the asymptotic law. In particular,
this argument does not prove that spatial restriction always removes the global
IR divergence; it only explains why the global logarithm does not by itself
force the same logarithm in spatial entropy.

\subsection{Relation to the spatial entropy problem}

The new scaling observables also clarify the role of the spatial
entropy. The overlap and Loschmidt amplitudes probe the global
change of the scattering vacuum directly. Spatial entropy first
applies a local restriction and only then performs a nonlinear
spectral functional. It is therefore perfectly consistent for
\begin{equation*}
|\langle\Psi_-|\Psi_+\rangle|\sim N^{-\gamma_{\rm AO}}
\end{equation*}
and
\begin{equation*}
|G(t)|\sim t^{-2\gamma_{\rm AO}}
\end{equation*}
while the spatial shutter correction remains parametrically smaller
on accessible intervals. The central distinction,
\begin{equation*}
\text{global infrared non-implementability}
\not\Rightarrow
\text{a determined spatial infrared entropy divergence},
\end{equation*}
is therefore strengthened rather than weakened by the new
calculations: the global infrared exponent is now measured
independently, rather than inferred from spatial entropy. The
implications of this separation for the spatial entropy program are
summarized in Sec.~\ref{sec:conclusions}.

\section{Discussion}
\label{sec:discussion}

The finite-width shutter gives a clean bridge between microscopic
scattering and Gaussian entanglement. The logical chain is
explicit: the static profile determines $S_a$, the switch
determines $\beta$ and $C_{\rm qp}$, the Gaussian state determines
$\Gamma$, and each bipartition applies its own projector before the
entropy is evaluated. At zero frequency, the profile enters only
through its integrated coupling and the scattering matrix is
exactly an $SO(2)$ rotation; at nonzero frequency the matrix is
generally complex and obeys the Majorana reality constraint rather
than literal $SO(2)$ reality. At finite frequency, the first
corrections are controlled by the dimensionless product
$a|\omega|$. This is the microscopic origin of the long-wavelength
universality used throughout the paper. The numerical analysis then
asks a more delicate question: how much of that one-particle
universality survives after restriction to a subsystem and
nonlinear processing of the covariance spectrum?

The entanglement problem is subtler. The two-chain and soft/hard
partitions explicitly retain the channel and frequency structure in
which the infrared memory is generated. Their finite-lattice
entropies are therefore naturally nonzero and show substantial
Rényi dependence. A continuum statement about their infrared
divergence, however, requires the singular-value distribution of
the reduced covariance, not merely the Hilbert--Schmidt norm of the
anomalous kernel.

The spatial partition is the sharper locality test. The vacuum has
the expected critical logarithm, but the shutter correction is
numerically much smaller on the accessible interval sizes. This
observation is compatible with the possibility that global infrared
memory is softened by spatial restriction. It is not a proof: the
accessible lattice is finite, and the sign and magnitude of the
correction vary with width and Rényi index.

The zero-memory qualification is equally important. The
inverse-frequency Bogoliubov kernel is an analytic result of the
instantaneous scattering approximation; the finite-memory
time-dependent scattering problem has not been solved here. The
small values $a/\tau=0.01$--$0.04$ in the controlled lattice runs,
together with the switching-time scan of
Sec.~\ref{sec:tau-convergence}, justify treating this as a
controlled regime parameter, but they do not by themselves furnish
a rigorous bound on the omitted continuum kernel correction.

The finite-width profile comparison also clarifies what
``universality'' should mean. The correct statement is not that
every finite lattice calculation gives identical entropy. The
microscopic theorem only says that smooth profiles approach the
same low-frequency scattering angle when their integrated coupling
is matched. Entanglement universality is a stronger property and
must be tested by a controlled scaling limit at fixed physical
angle.

The conformal-defect comparison is now made in the geometry for which the
known defect formula applies: the entanglement cut passes through the shutter.
Using the parameter-free transmission $s_{\rm shutter}=\sech\lambda_f$, the
half-chain entropy slope agrees with the Eisler--Peschel effective-central-charge
prediction at the percent level over the tested range. This benchmark is static
and does not determine the entropy of an interval whose cuts lie away from the
defect.

The present results also motivate a clear division between
established and open statements, summarized in
Sec.~\ref{sec:what-established}. The exact zero-frequency transfer
matrix and the long-wavelength scattering limit are analytic
results. The stability of the leading infrared coefficient of the
anomalous Hilbert--Schmidt norm is established within the
zero-memory approximation. The finite-lattice entropies are
numerical pilot results. The continuum entropy asymptotics and the
exact finite-memory correction are open calculations.

A complexity observable is not included. It would require
additional choices of reference state, circuit class, metric, and
continuum renormalization. Those choices are not needed to answer
the present, more sharply defined question: whether a microscopic
finite-width shutter has a universal entanglement response.

\section{Conclusions}
\label{sec:conclusions}

We have developed a finite-width Majorana shutter framework in which the
microscopic scattering problem, Gaussian many-body response, and spatial
entanglement are treated within a common regulator. The central analytic
result is the exact zero-frequency scattering matrix
\begin{equation*}
S_a(0)=R[\arctan(\sinh\lambda)],
\end{equation*}
which is independent of the detailed normalized smooth profile. At finite
frequency the approach to this rotation is controlled by $a|\omega|$ within
the stated profile class, and the finite-width coupling is related to the
symmetric point-interaction convention by
$\lambda=2\operatorname{artanh}(g/2)$.

Within the zero-memory approximation, a time-dependent switch inherits the
same low-frequency scattering data and produces an inverse-sum-frequency
Bogoliubov kernel. The squared Hilbert--Schmidt norm therefore has a logarithmic
infrared growth. The finite-memory correction has not been derived or bounded
rigorously, so this continuum statement is explicitly conditional on the
zero-memory regime.

The lattice calculations give three independent many-body manifestations of
the scattering angle. First, the doubled-lattice ground-state overlap obeys
$|\langle\Psi_-|\Psi_+\rangle|\sim N^{-\gamma_{\rm AO}}$ with
$\gamma_{\rm AO}=(\Delta\theta/2\pi)^2$ within the quoted finite-width
and profile corrections; a single undoubled continuum Majorana pair carries
one half of this exponent. For the open-chain regulator the factor of two is
validated empirically by the agreement with the parameter-free AO and excitation
coefficients; it is not claimed as a first-principles boundary-exponent derivation.
The Loschmidt amplitude in the pre-revival window
is consistent with the corresponding factor-of-two time exponent.

Second, the post-quench excitation number
$N_{\rm exc}=\|\beta\|_{\rm HS}^2$ grows logarithmically with system size.
For the doubled lattice its measured coefficient follows
\begin{equation*}
2C_{\rm qp}=\frac{4}{\pi^2}
\sin^2\left(\frac{\Delta\theta}{2}\right),
\end{equation*}
for the tested sudden switches and for the smooth $\tau=10$ and $\tau=100$
protocols. The $\tau=100$ calculation reaches
$\varepsilon_1\tau\simeq0.10$ at $N=3072$ and gives a local slope of
$0.0708$, compared with the prediction $0.0713$; a fit over
$N=1536$--$3072$ gives $0.0714$. This is a ground-state finite-size scaling measurement of the coupled lattice,
not a dynamical measurement of the continuum Bogoliubov kernel and not a
continuum asymptotic extrapolation. Its agreement with the parameter-free
prediction is nontrivial, but it does not by itself establish the zero-memory
approximation. At weak angle the excitation
coefficient reduces to $4\gamma_{\rm AO}$, whereas at finite angle the two
functions differ because one is proportional to $\sin^2(\Delta\theta/2)$.

Third, when the entanglement cut passes through the shutter, the half-chain
entropy follows the conformal-defect effective central charge evaluated at the
parameter-free transmission
$s=\sech\lambda_f=\cos\theta_{\lambda_f}$. The measured logarithmic
slopes agree with the Eisler--Peschel prediction at the percent level over the
tested transmission range. This result is specific to the defect-cut geometry:
when the shutter lies inside an interval whose cuts remain in the homogeneous
bulk, the present data show only a subleading finite-size correction and no
evidence for a modified leading logarithm.

For the deterministic numerical results, the quoted spreads should be read as
systematic envelopes rather than statistical uncertainties. The dominant sources
are finite-size drift, fit-window dependence, and finite-width/profile effects;
the time-step refinement is below the percent level for the controlled entropy
run, while the AO exponent varies by about $\pm7\%$ across the tested
width/profile/fit choices and the Loschmidt slope varies by about $4\%$ over
the stated fit windows. No controlled $1/N$ AO extrapolation is reported because
the retained data record does not contain the raw overlap amplitudes needed to
separate finite-size drift from fit-window dependence.

The local bilinear two-point function provides a separate consistency check.
After connected subtraction and treatment of the doubled-lattice staggered
component, its fitted exponent approaches the ordinary Majorana value $p=2$,
while the fitted logarithmic amplitude is not stable under the available
regulator and control variations. The data therefore provide no evidence for a
logarithmic partner in this tested correlator, without making a claim about the
full scaling-limit operator algebra.

The remaining limitations are sharply defined. The finite-memory correction to
the time-dependent continuum kernel remains an open analytic problem. The
soft/hard frequency-partition results are finite-lattice diagnostics and do not
establish fixed-physical-cutoff continuum convergence. Most importantly, the
continuum large-interval spatial entropy has not been assigned a universal
coefficient. Establishing such a result would require controlled scaling of the
restricted covariance at fixed physical scattering angle, together with the
appropriate ultraviolet and infrared limits.

The resulting hierarchy is therefore
\begin{equation*}
\boxed{
\Delta\theta
\longrightarrow
\{\gamma_{\rm AO},\;C_{\rm qp},\;\text{Loschmidt exponent},\;c_{\rm eff}(s)\}
}
\end{equation*}
with the important qualification that these quantities describe different
observables and are not interchangeable. The finite-width shutter thus provides
a concrete setting in which microscopic low-frequency scattering data can be
followed into several quantitatively testable many-body responses without
assuming a universal spatial entropy law that has not been established.

\section*{Data and code availability}
The numerical data and analysis scripts used to generate the
tables, figures, overlap and Loschmidt fits, LCFT correlator
controls, anchor-site scans, fit-window checks, and the numerical
diagnostics discussed in this paper are provided with the
manuscript source package. The non-adiabatic chain and soft/hard
diagnostics are supplied separately as supplementary material.


\appendix

\section{Transfer-matrix details}
\label{app:transfer}

For completeness, the amplitude equation can be written as
\begin{equation*}
\partial_xY=N_a(x;\omega)Y,
\qquad
Y=(A,B)^T.
\end{equation*}
The antilinear map $Y\mapsto\sigma_1\overline Y$ preserves the
equation. Hence
\begin{equation*}
T_a\sigma_1=\sigma_1\overline{T_a},
\end{equation*}
which gives the reality structure used in Eq.~\eqref{eq:SfromT}.
Since $\Tr N_a=0$, Liouville's formula gives $\det T_a=1$.

At $\omega=0$,
\begin{equation*}
\mathcal P\exp\left[\lambda\sigma_1\int w_a\right]
=e^{\lambda\sigma_1}.
\end{equation*}
For nonzero $\omega$, the Dyson expansion is controlled by the
support of $w_a$. Each insertion carries a phase difference of
order $a|\omega|$ relative to the zero-frequency generator. The
first moment of an even profile vanishes, and the remaining ordered
terms give the bound in Eq.~\eqref{eq:Tbound}.

\section{Self-contained derivation of the point-contact infrared logarithm}
\label{app:IR-derivation}

The infrared result used in Sec.~\ref{sec:IR} is that the squared
Hilbert--Schmidt norm of the anomalous Bogoliubov kernel grows
logarithmically with the infrared cutoff. We reproduce the
derivation here for completeness.

Consider a point-contact shutter with a smooth switch
$\lambda(t)\to\lambda_\pm$ as $t\to\pm\infty$. The scattering
matrix at zero frequency is $S_\pm=R(\theta_\pm)$, and the
anomalous kernel is
\begin{equation}
\beta(\omega,\nu)=\frac{1}{2\pi}\int_{-\infty}^\infty dt\,
e^{i(\omega+\nu)t}\,S(t;\omega),
\qquad \omega,\nu>0.
\end{equation}
For small positive frequencies, the dominant contribution comes
from the asymptotic difference $\Delta S=S_+-S_-$:
\begin{equation}
\beta(\omega,\nu)\sim-\frac{\Delta S}{2\pi i(\omega+\nu)}
\qquad (\omega,\nu\to0^+).
\end{equation}
The squared Hilbert--Schmidt norm, integrated over a frequency
window $\mu<\omega,\nu<\omega_*$, is
\begin{equation}
\|\beta_{\mu,\omega_*}\|^2_{\rm HS}
=\frac{\|\Delta S\|^2_{\rm F}}{4\pi^2}
\int_\mu^{\omega_*}d\omega\int_\mu^{\omega_*}d\nu\,
\frac{1}{(\omega+\nu)^2}+\cdots .
\end{equation}
The integral is elementary:
\begin{equation}
\int_\mu^{\omega_*}d\omega\int_\mu^{\omega_*}d\nu\,
\frac{1}{(\omega+\nu)^2}
=\log\frac{\omega_*}{\mu}+O(1).
\end{equation}
Therefore
\begin{equation}
\|\beta_{\mu,\omega_*}\|^2_{\rm HS}
=\frac{\|\Delta S\|^2_{\rm F}}{4\pi^2}\log\frac{\omega_*}{\mu}+O(1),
\end{equation}
which is Eq.~\eqref{eq:HSlog}. The coefficient $C_{\rm qp}$ in
Eq.~\eqref{eq:Cqp} follows from
$\|\Delta S\|^2_{\rm F}=8\sin^2(\Delta\theta/2)$ for rotations.
This derivation is self-contained and does not rely on the
companion paper \cite{Vahedi2026a}; the companion result is used
only for the point-contact statement that the same logarithm
appears in the genuinely time-dependent point-contact problem.

\section{Gaussian entropy and canonical form}
\label{app:entropy}

Every real antisymmetric covariance matrix admits an orthogonal
canonical form
\begin{equation*}
O_A\Gamma_AO_A^T
=\bigoplus_j
\begin{pmatrix}0&\xi_j\\-\xi_j&0\end{pmatrix}.
\end{equation*}
Each block corresponds to one effective fermionic mode with
reduced density matrix
\begin{equation*}
\rho_j=
\begin{pmatrix}
(1+\xi_j)/2&0\\
0&(1-\xi_j)/2
\end{pmatrix}
\end{equation*}
up to the ordering convention for the occupation basis.
Therefore
\begin{equation*}
\rho_A\simeq\bigotimes_j\rho_j,
\end{equation*}
which immediately yields Eqs.~\eqref{eq:svN} and \eqref{eq:Renyi}.
No diagonalization of a reduced density matrix of exponential
dimension is required.

The same construction applies to spatial, two-chain, and
energy-mode partitions. The only change is the projector $P_A$
selecting the subsystem Majoranas.

\end{document}